\documentclass[oupdraft]{bio}

\usepackage{url}

\usepackage[T1]{fontenc}    

\usepackage{url}            
\usepackage{booktabs}       
\usepackage{amsfonts}       
\usepackage{nicefrac}       
\usepackage{microtype}      
\usepackage{lipsum}
\usepackage{enumerate}      
\usepackage{rotating}       

\usepackage{amsmath}
\usepackage{amsthm}
\usepackage{mathrsfs}
\usepackage{color,soul} 
\usepackage{graphicx}

\usepackage{multirow}
\usepackage{makecell}

\theoremstyle{definition}
\newtheorem{corollary}{Corollary}

\title{MM-PCA: Integrative Analysis of Multi-group and Multi-view Data}

\author{JONATAN KALLUS\\[4pt]
\textit{Mathematical Sciences, University of Gothenburg and Chalmers University of Technology, Gothenburg, Sweden}
\\[2pt]
PATRIK JOHANSSON, SVEN NELANDER\\[4pt]
\textit{Department of Immunology, Genetics and Pathology, Uppsala University, Uppsala, Sweden}
\\[2pt]
REBECKA J\"ORNSTEN$^{\ast}$\\[4pt]
\textit{Mathematical Sciences, University of Gothenburg and Chalmers University of Technology, Gothenburg, Sweden}
\\[2pt]
{jornsten@chalmers.se}
}

\begin{document}

\maketitle
\footnotetext{To whom correspondence should be addressed.}

\begin{abstract}
{Data integration is the problem of combining multiple data groups (studies, cohorts) and/or multiple data views (variables, features). This task is becoming increasingly important in many disciplines
due to the prevalence of large and heterogeneous data sets. 
Data integration commonly aims to identify structure that is consistent across multiple cohorts and feature sets.
While such joint analyses can boost information from single data sets, it is also possible that a globally restrictive integration of heterogeneous data may obscure signal of interest. Here, we therefore propose a data adaptive integration method, allowing for structure in data to be shared across an a priori unknown \emph{subset of cohorts and views}. The method, Multi-group Multi-view Principal Component Analysis (MM-PCA), identifies partially shared, sparse low-rank components. This also results in an integrative bi-clustering across cohorts and views. 
The strengths of MM-PCA are illustrated on simulated data and on 'omics data from The Cancer Genome Atlas. MM-PCA is available as an R-package. }
{Data integration \and Multi-view \and Multi-group \and Bi-clustering}
\end{abstract}

\section{Introduction}
Large-scale studies commonly involve multiple sources of data. These data sets can stem from in-house experiments, public data-base resources,  studies conducted in different experimental settings or obtained through different technological platforms. There is thus an ever increasing need to provide analysis methods that can integrate heterogeneous data from several cohorts or studies and several sets of features. 

\begin{figure}
  \centering
    \includegraphics[scale=0.4]{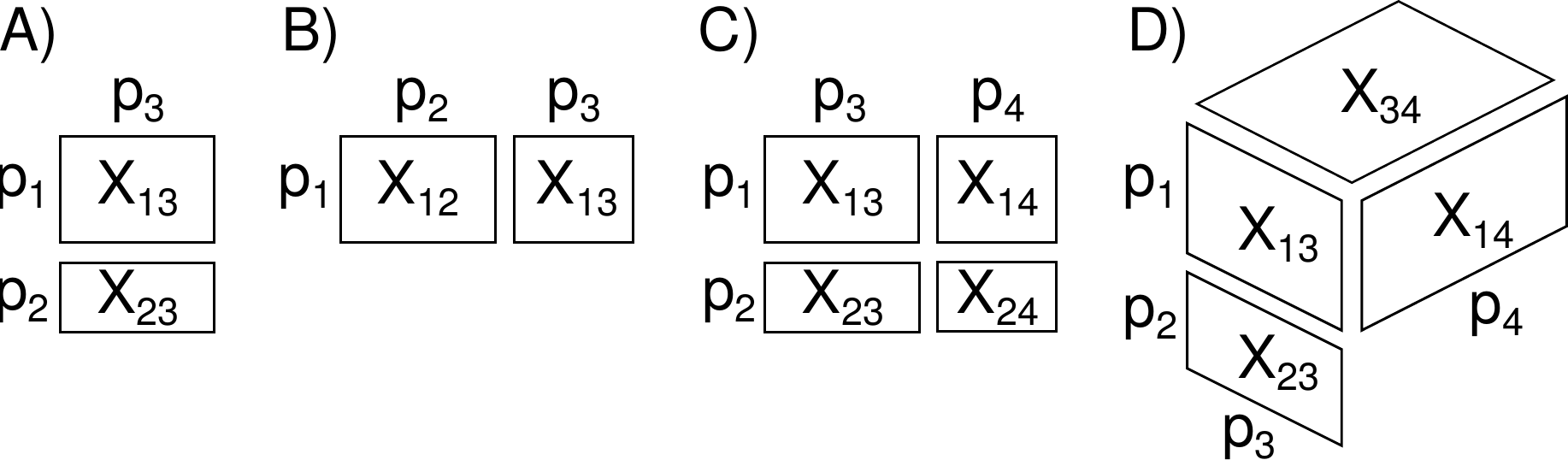}
  \caption{Examples of A) multi-group, B) multi-view, C) multi-group and multi-view, and D) augmented multi-view data. $p_i$ is the dimensionality of view $i$ and $X_{ij}$ is the data matrix that contains data for views $i$ and $j$.}
  \label{fig:amd}
\end{figure}

The integration of data from multiple cohorts observed across a common set of features is commonly referred to as multi-group integration (Figure~\ref{fig:amd}A). Similarly, the multi-view (multi-source, multi-modal) data integration problem focuses on one cohort where multiple sets of features or measurements are available (Figure~\ref{fig:amd}B). The multi-group and multi-view settings are to some extent equivalent algorithmically. By transposing the data matrices of a multi-group data set, it can be analyzed by multi-view methods and vice versa.

Several methods for multi-view (or multi-group) integrative data analysis are based on PCA or SVD. 
Here, we provide a brief summary while an extensive literature review can be found in the online supplement (Table~\ref{tab:review}). 
Methods that extend to more than two matrices  are generally limited to finding a joint structure shared by all matrices. 
The interpretational value of multi-group (or multi-view) data integration can be improved by also modeling properties that are only present in one data matrix. \cite{van_der_kloet_separating_2016} present a comparative review of such methods, including JIVE \citep{lock_joint_2013}, DISCO-SCA \citep{schouteden_sca_2013} and O2-PLS \citep{trygg_o2-pls_2002},  in the context of multi-view data. The three methods
all fit a global model for the joint structure and individual models to the residuals in each matrix. They differ primarily in how strongly they impose assumptions of orthogonality of the individual structures of each matrix. 
Recently, \cite{feng_angle-based_2018} proposed AJIVE which uses the same model as JIVE but with improved computational efficiency for parameter estimation. \cite{zhu_generalized_2018} developed GIPCA: a generalization of JIVE to different data types which also allows for the presence of missing values. 

To enhance interpretability further still, sparsity constraints can be applied to the structure parameterizations (e.g.\  the loadings of an SVD-based analysis). For example, \cite{wang_sparse_2015} proposed sMVMF to provide a decomposition of multi-group data into global and individual signal while imposing sparsity penalties on the global and individual loadings. This results in a stratification of features into sets that are; (i) globally relevant to explain all data groups; (ii) individually relevant; and (iii) irrelevant. 

Simultaneous data integration for multi-group \emph{and} multi-view data (Figure~1C) was first addressed by \cite{lofstedt_bi-modal_2012}. The method, bi-modal OnPLS,   
can decompose such datasets into global and individual signal. Linked matrix factorization by \cite{oconnell_linked_2017} later extended JIVE to the same generality, with the additional ability to handle elementwise missing data. BIDIFAC \citep{park_integrative_2019} is another extension of JIVE which, in addition, allows signal to be group-global and view-global.

In this paper, we aim to generalize the multi-group and multi-view integration to allow information to be only partially shared across a subset of groups and/or a subset of views. 
In many applications, a dataset may contain heterogenous matrices, and it is therefore unlikely that a globally shared signal exists.
In such datasets, methods that only look for globally joint signal are forced to explain all variation as individual.

\cite{argelaguet_multiomics_2018} recently developed MOFA, which is, to our knowledge, the only method for multi-view data that can find joint signal shared only for a subset of matrices. \cite{klami_group-sparse_2014} proposed CMF to perform multi-group and multi-view integration with loadings that are shared by a subset of matrices. Both MOFA and CMF utilize a Bayesian setup, where an automatic relevance determination (ARD) prior is used for each loading so that some are active for only a subset of its associated data matrices. The ARD prior does not generate exact zeros. Consequently, it is difficult to interpret results in terms of which data matrices have a common structure, and how that common structure should be characterized. Furthermore, it does not favor low rank models. CMF does not produce sparse loadings. MOFA combines an ARD prior with a spike-and-slab prior where the role of the spike-and-slab prior is to produce sparse loadings (feature selection) but this setup has not been generalized to a multi-group and multi-view setting.

We propose a frequentist multi-group and multi-view data integration method, MM-PCA, that utilizes a symmetric treatment of views and groups, paired with regularization, to achieve; (i) partial component sharing across subsets of groups and/or views; (ii) orthogonal components through an efficient parametrization; (iii) sparse loadings and scores with respect to each component; and, (iv) adaptive rank selection for the integrative procedure. Comparing MM-PCA to CMF, CMF offers an approximate solution to (i) and does not provide solutions for (ii-iv) which greatly reduces the interpretability of the CMF data integration procedure. 

MM-PCA can be used for explorative analysis, for example integrative bi-clustering and the identification of key features that drive cohort stratification. In integrative genomics, such driver-strata components can have relevance for development of new therapeutic solutions. To demonstrate the strength and usefulness of MM-PCA we conduct simulations and apply MM-PCA for the analysis of data from the Cancer Genome Atlas.

The rest of this article is organized as follows. Section \ref{sec:amvpca} defines the model used in MM-PCA and how its parameters are estimated and should be interpreted. 
In Section \ref{sec:sim} we compare MM-PCA with CMF, the only existing method of similar generality, and with JIVE, a popular method for integrative analysis in genomics, in three simulation based experiments. In Section \ref{sec:bio} we apply MM-PCA to a data set from the Cancer Genome Atlas (TCGA, \cite{tcga_stuart_2013}).
Section \ref{sec:disc} concludes the article and discusses limitations and potential future improvements.

\section{Multi-group and Multi-view Principal Component Analysis}
\label{sec:amvpca}

To facilitate the generalization of multi-group and multi-view integration, we treat the data integration as symmetric with respect to group or view. \cite{klami_group-sparse_2014} termed this approach augmented multi-view data (AMD). To define AMD, let us consider $n_m$ data matrices. Two matrices can be directly related by sharing rows or columns, or indirectly related via other matrices. More precisely, two matrices share rows if row $i$ in each of the two matrices concern the same object, and similarly for columns. This defines a partition of the $2n_m$ sets of rows/columns into $n_v$ sets. We call such sets \emph{views}, regardless of whether they consist of observations or features. 
For example, consider two data cohorts; cohort 1 with $p_1$ patients and cohort 2 with $p_2$ patients. Each cohort has two feature sets; 
the expression of $p_3$ genes and the methylation state of $p_4$ probes (Figure \ref{fig:amd}C). There are 8 sets of rows/columns for 4 unique views. The generality of AMD allows for the inclusion of a $p_3\times p_4$ matrix with prior knowledge of covariance between gene expression and methylation. It also allows for inclusion of feature sets for which data are not available for all cohorts (Figure \ref{fig:amd}D).

Let $X_{ij}~((i,j)\in\mathcal{S})$ be a $p_i\times p_j$ data matrix, where $\mathcal{S}$ is the set of view pairs $(i,j)$ for which a data matrix exists. The cardinality $|\mathcal{S}|=n_m$ is the number of matrices in the data set in total. MM-PCA aims to find a rank $k$ integrative solution to minimize the following loss
\begin{align}
\label{eq:mmpcaeq}
\min_{D_\cdot,V_\cdot} \sum_{(i,j)\in\mathcal{S}}||X_{ij}-V_iD_iD_jV_j^T||_F^2,
\end{align}
subject to $V_i^TV_i=I_k,i=1,\ldots,n_v$. $V_i$ ($p_i\times k$) are the view-specific loading matrices and $D_i$ are diagonal $k \times k$ matrices. 

The optimal MM-PCA decomposition of a single data matrix $X_{12}$ is $V_1D_1D_2V_2^T$, where $V_1D_1D_2$ is equal to the standard PCA scores for $X_{12}$ and $V_2$ is equal to the PCA loadings. Since MM-PCA treats rows and columns symmetrically, we call all $V_\cdot$ matrices loadings.

Without further constraints each column of $V_i$ would influence all matrices $X_{i\cdot}$. 
By imposing exact zeros on some elements of the  $D_\cdot$ matrices we can constrain the MM-PCA components to be shared between only a subset of data matrices. 
Since the $D_\cdot$ matrices are diagonal and of the same size, $k \times k$, they can  be compactly represented by a $k\times n_v$ matrix $\mathscr{D}$ where each column is equal to the diagonal of one of the $D_\cdot$ matrices and each row corresponds to a MM-component. We call $\mathscr{D}$ the \emph{augmented D matrix}. 

We now illustrate how a  \emph{joint structure} is defined by $\mathscr{D}$. A zero at the element corresponding to component $c \in \{1,\ldots, k\}$ (row) and view $j$ (column) means that component $c$ is not capturing variation in data matrices with data for view $j$. Thus, all views with non-zero values for component $c$ contain variation that is jointly captured by component $c$. For example, consider matrices $X_{12}$, $X_{13}$ and $X_{14}$. If
$\mathscr{D}_{c,2}=0$ and $\mathscr{D}_{c,j}\neq 0, j=1,3,4$,
then component $c$ would be shared between matrices $X_{13}$ and $X_{14}$ but not $X_{12}$. If several rows of $\mathscr{D}$ have identical patterns of zeros, then the corresponding components make up a subspace
where the views with non-zero values co-vary.  

$\mathscr{D}$ defines both the proportion of variation explained by estimated components and how much of the explained variation that is shared between data matrices. Denote the optimal $k$-component MM-PCA solution for matrices $X_{ij}, (i,j) \in \mathcal{S}$  by 
$\mathcal{M}_{k,n_v}=\{\mathscr{D}_{k \times n_v}, V_j, j=1,\ldots,n_v\}$. The vector $\mathscr{D}_{ci}\mathscr{D}_{cj},c\in \{1,\ldots, k\}$ is the vector of singular values of the approximation of $X_{ij}$. Thus, the proportion of variation explained by the $k$-rank solution $\mathcal{M}_{k,n_v}$ of $X_{ij}$ is given by 
$$
R^2(X_{ij} \mid \mathcal{M}_{k,n_v}) =  ||V_iD_iD_jV_j^T||_F^2/||X_{ij}||_F^2=\sum_{c=1}^k\mathscr{D}_{ci}^2\mathscr{D}_{cj}^2/||X_{ij}||_F^2.
$$
Similarly, the proportion of variation in a matrix $X_{ij}$ that can be linearly predicted with another matrix $X_{i'j'}$ can be computed as follows: 
$$R^2(X_{ij} \mid \mathcal{M}_{k,n_v} , X_{i'j'}) = 
\sum_{c\in\mathcal{C}}\mathscr{D}_{ci}^2\mathscr{D}_{cj}^2/||X_{ij}||_F^2, \ \ \mathcal{C} = \{c: \mathscr{D}_{ci'}\mathscr{D}_{cj'}\ \neq 0\},$$
where $\mathcal{C}$ is the set of components that influence $X_{i'j'}$. 
This gives a directed relation between each pair of data matrices, and can be used to form hypotheses of how groups of variables influence each other (Figure \ref{fig:bioex}B).

\subsection{Estimation}
\label{sec:estimation}
We estimate the parameters $V_\cdot$ and $D_\cdot$ using quasi-Newton local optimization of the sum of the objective \eqref{eq:mmpcaeq} and four penalty functions (to be defined in Section \ref{sec:obj}).

Two problems arise ; 
(i) for relevant biological datasets the resulting optimization problem involves a massive amount of parameters (order of $10^4$ to $10^5$); and (ii) the constraints are challenging to handle numerically. Massive amount of parameters is not an insurmountable problem in optimization, as seen for example in recent advances in the field of deep neural networks. The constraints are challenging, however, in that they are non-linear and constrain the set of feasible solutions to a manifold so that all solutions are on the border of the set. 

To address these problems, we introduce a parsimonious parametrization for the loading matrices. This not only reduces the number of parameters, but, more importantly, also removes the optimization constraints. 

\subsubsection{Parsimonious Parametrization}
\label{sec:geneuler}

Let $V$ (e.g. $V_1$) be $k$ columns of a $p$-dimensional rotation matrix, i.e.\ it is a matrix of $pk$ parameters constrained to have orthogonal columns of unit $\ell_2$-norm. Such matrices are called \emph{k-frames} and effectively have $pk-k(k+1)/2$ free parameters \citep{james_normal_1954}. 

To find a parsimonious parametrization of $V$ we make use of a parametrization based on Givens rotation matrices \citep{shepard_representation_2015}.
A rotation matrix for two-dimensional rotation can be derived from angles, $\theta$.
Generalizing this definition to $p$-dimensional space requires the multiplication of several rotation matrices, each constituting a two-dimensional rotation in a two-dimensional subspace. Such matrices \begin{align*}
R_{ij}(\theta)=\begin{bmatrix}I_{i-1} & 0 & 0 & 0 & 0\\
0 & \cos\theta & 0 & -\sin\theta & 0\\
0 & 0 & I_{j-i-1} & 0 & 0\\
0 & \sin\theta & 0 & \cos\theta & 0\\
0 & 0 & 0 & 0 & I_{p-j}\end{bmatrix}
\end{align*}
are known as Givens rotation matrices.

We include in the supplement a proof in our notation that for $p\geq k$, an arbitrary $p$-dimensional $k$-frame $V$ can be expressed as a multiplication of $m=pk-k(k+1)/2$ Givens rotation matrices as 
$V(\xi_{\cdot})=(\prod_{i=1}^p\prod_{j>i}^k R_{ij}(\xi_{\cdot,ij}))I_{pk}$,
where $I_{pk}$ are the first $k$ columns of the $p$-dimensional identity matrix. 
Each element $\xi_{\cdot,ij} $ is an angle, and elements are collected in  
$\xi_{\cdot}$, a lower triangular matrix with zero diagonal, of dimension $p \times k$. The function $V(\cdot)$
thus yields a parsimonious parametrization of k-frames from $m$ angles.

This parametrization was recently used by \cite{pourzanjani_general_2017} for Bayesian inference of probabilistic PCA to facilitate sampling from a distribution over k-frames. It has not previously been used for optimization based inference of SVD nor in the context of SVD based data integration, and we have adapted the framework to more conveniently be applied for such tasks.

With this parametrization of k-frames we can restate \eqref{eq:mmpcaeq} without the challenging orthogonality constraints on the loadings.
That is, we write for view $i$, $V_i=V(\xi_i)$, 
$\xi_i$
of sizes $p_i\times k$. 
When computing $V$, the product of several large matrices, we make use of the sparse and regular structure of Givens rotation matrices.

We use the inverse of $V(\cdot)$, stated in the supplement, to find initial parameter values for the optimization problem \eqref{eq:mmpcaeq}.

\subsubsection{Loss, penalties and model selection}
\label{sec:obj}
We now restate the loss function as
\begin{align*}
\sum_{(i,j)\in\mathcal{S}}||X_{ij}-V(\xi_i)D_iD_jV(\xi_j)^T||_F^2 + \sum_{c=1}^4 \lambda_c\text{P}_c(\xi,D),
\end{align*}
with four added penalty terms to achieve; (i) data integration through the identification of partially shared components; (ii) rank selection; (iii) sparse loadings; and (iv) variable selection. 

Data integration  is achieved via an $\ell_1$ sparsity penalty on the diagonals of $D_\cdot$: $P_1(\cdot,D) = \lambda_1\sum_{i=1}^{n_v}||D_i||_1$, where $n_v$ is the number of views and $||\cdot||_1$ is the sum of the absolute values of the matrix' elements. This penalty plays a similar role to the ARD prior on loadings used in CMF. Unlike CMF, however, we have separated the loadings ($V_\cdot$) from their norms ($D_\cdot$). This allows us to define an integration penalty that only involves a small fraction of the parameters. 

The rank of solutions is penalized by adding a group penalty on each diagonal position across all views: $P_2(\cdot,D) =  \lambda_2\sum_{c=1}^k\sqrt{\sum_{i=1}^{n_v}(D_i)_{cc}^2}$.

Sparsity of loadings is also achieved with an $\ell_1$ penalty. To avoid solutions where less important components tend to be more sparse, each $V$ is multiplied by its associated $D$: $P_3(\xi,D) = \lambda_3n_v^{-1}\sum_{i=1}^{n_v}||V(\xi_i)D_i||_1$. 

Variable selection is achieved with a group penalty on each row of the loading matrices: $P_4(\xi,D)=\lambda_4n_v^{-1}\sum_{i=1}^{n_v}\sum_{d=1}^{p_i}||(V(\xi_i)D_i)_{d\cdot}||_2$, where $p_i$ is the dimension of view $i$. 

The division by $n_v$ is in both cases used to make all hyper parameters have approximately the same scale.

To allow for missing elements in the data matrices, the corresponding terms of the loss function are removed. For each data matrix $X_{ij}$ define a matrix $M_{ij}$ whose elements are zero if the same position in $X_{ij}$ is missing and one otherwise. The MM-PCA loss function thus becomes
\begin{align}
\begin{split}
\label{eq:objective}
\sum_{(i,j)\in\mathcal{S}}||M_{ij}\odot\left(X_{ij}-V(\xi_i)D_iD_jV(\xi_j)^T\right)||_F^2+\lambda_1\sum_{i=1}^{n_v}||D_i||_1+\\
+\lambda_2\sum_{c=1}^k\sqrt{\sum_{i=1}^{n_v}(D_i)_{cc}^2}+\frac{\lambda_3}{n_v}\sum_{i=1}^{n_v}||V(\xi_i)D_i||_1+\frac{\lambda_4}{n_v}\sum_{i=1}^{n_v}\sum_{d=1}^{p_i}||(V(\xi_i)D_i)_{d\cdot}||_2,
\end{split}
\end{align}
where $\mathcal{S}$ is the set of pairs $(i,j)$ for which a data matrix $X_{ij}$ exists with rows concerning view $i$ and columns concerning view $j$ and $\odot$ is elementwise multiplication. This optimization problem has $n_vk+\sum_{i=1}^{n_v}(p_ik-k(k+1)/2)$ variables.

The  $D_\cdot$ matrices only enter the loss function in pairs, and thus  $D_iD_j, (i,j)\in\mathcal{S}$ (Figure \ref{fig:R2}B) are more directly informative than the $D_\cdot$ themselves (Figure \ref{fig:R2}A). The former contain the joint structure, the amount of variation explained and the amount of shared variation between matrices (as defined above). Still, values in $D_\cdot$ are not arbitrary. The penalty terms promote solutions where the elements of the diagonal matrices are of similar size, and each $D_\cdot$ may partake in several matrices.

Optimization of the objective function \eqref{eq:objective} is performed with the Broyden-Fletcher-Goldfarb-Shanno (BFGS) algorithm, a quasi-Newton method that makes use of the gradients of the objective function (gradients shown in the supplement) for iteratively updating an estimate of the Hessian. 
Details on how the various penalties are scaled to remove their dependencies on scale and size of the data matrices are given in the supplement. 

Values for the penalty parameters $\lambda=(\lambda_1,\ldots,\lambda_4)$ are chosen by dividing the elements of the data matrices into a test set and a training set. Non-missing elements of the data matrices are randomly assigned to the test set (with probability 0.1). For given values of $\lambda$, the model is fitted with the elements of the test set treated as missing data. A $\lambda$ is sought that minimizes the squared reconstruction error of elements in the test set. The user can specify any sequence of candidate $\lambda$ values. To limit the search space we typically choose candidate values by deciding which penalties that should be active (specified as four binary parameters $(b_1,\ldots,b_4)$) and choosing a range and step size of values $\lambda_0$. The set of candidate values is then given by $\lambda=(\lambda_0b_1,\ldots,\lambda_0b_4)$ for each value of $\lambda_0$. After the optimal $\lambda$ has been found, the model is finally estimated using all data (including the test set). Each candidate value of $\lambda$ can be tried in parallel, making it possible to utilize many processor cores simultaneously for parameter tuning.

\subsubsection{Normalization and Initialization}

It is advisable to center both rows and columns of each $X$ similarly to normalization for ordinary PCA.
One may also normalize data matrices in relation to each other through scaling e.g.\ to equal Frobenius norm, (square) Frobenius norm proportional to the number of elements, rows or columns, or equal Frobenius norm of the first principal components. Further matrix-wise normalizations are discussed in Section 2.5 of \cite{lofstedt_bi-modal_2012}.
In the end, data normalization must depend on data or domain knowledge.

We initialize the optimization with an approximation of the solution under the assumptions that all components are globally joint. Alternatively, MM-PCA can use CMF to find initial values. The R package allows for both.

\subsection{Components of the MM-PCA solution}
To summarize the above sections, MM-PCA provides a multi-group and multi-view  structure decomposition of data matrices that includes; 
\begin{enumerate}[(i)]
    \item \textbf{Automatic rank selection.}~~MM-PCA estimates the optimal rank by using a group penalty on the rows of $\mathscr{D}$. The estimated rank is given by the number of rows in $\mathscr{D}$ that contain non-zero elements.
    \item \textbf{Interpretable sparse loadings.}~~Sparsity in loadings, that is in columns of $V_\cdot$, increases interpretability.
    \item \textbf{Multi-group multi-view bi-clustering.}~~
    Sparse loadings facilitate a direct translation from the MM-PCA decomposition to an integrative bi-clustering of the data. Each column of the loading matrix partitions the observations or variables of a view into three groups: positive, zero and negative. The use of sparse loadings for bi-clustering of a single data matrix was explored in \cite{lee_biclustering_2010} which, with MM-PCA, is thus generalized to the multi-group and multi-view setting.
    \item \textbf{Detection of outliers and irrelevant variables.}~~Removing outliers (row or columns) from the model is achieved via a group penalty on rows of $V_\cdot$. This penalizes difference in sparsity structure between loadings for the same view. Removing an outlier means that the corresponding row or column is considered to contain noise only.
    \item \textbf{Imputation of missing values and missing matrices of data.}~~
    Data matrix approximations given by the fitted model have no missing values, thus achieving imputation. Furthermore, data matrices that are not present when fitting the model, but that has views that are present through other matrices, can be predicted. E.g.\ if $X_{ij'}$, $X_{i'j}$ and $X_{i'j'}$ are present, a missing matrix consisting of views $i$ and $j$ is predicted by $\hat{X}_{ij}=V_iD_iD_jV_j^T$.
\end{enumerate}

\section{Method Evaluation on Simulated Data}
\label{sec:sim}
We investigate MM-PCA performance via three simulation studies. Simulation 1 is aimed at demonstrating MM-PCA's ability, compared with CMF, to correctly identify partially shared structure in a data set. Simulation 2 investigates MM-PCA's automatic rank selection and ability to identify globally joint signal, where we compare against CMF and JIVE.  Finally, simulation 3 demonstrates MM-PCA's ability to identify sparse loadings.

\subsection{Simulation 1} 
We generate four data matrices of size $10\times 10$, thought to represent four different cohorts (views 1 through 4) observed over a common set of features  (view 5). 
Each matrix is the sum of a rank one matrix (the signal) and noise.  Data matrices one and two share a rank 1 component, and similarly for matrices three and four.
That is, 
\begin{align}
\begin{split}
\label{eq:sim1}
X_i&=u_iv_i^T+c\varepsilon_i,~i=1,\ldots,4, 
\hspace{.5cm}
u_i^T\sim N(0, I)\\
v_i&=\tilde{v}_i/|\tilde{v}_i|_2,
\tilde{v}_1^T=\tilde{v}_2^T\sim N(0, I),
\tilde{v}_3^T=\tilde{v}_4^T\sim N(0, I)
\end{split}
\end{align}
where each element of each $\varepsilon_i$ is independent standard normal and $c$ is chosen to obtain a fixed signal-to-noise ratio (SNR).

Both MM-PCA and CMF are run with default settings. The correct noiseless data rank (two) is given as input to both methods. This gives an advantage to CMF since it does not penalize the rank of the estimated model.

We compare  the estimated joint structure to the correct structure. Recall from Section \ref{sec:amvpca}, the structure of the decomposition is captured by the augmented $D$-matrix, $\mathscr{D}$. For row (component) $k$ of $\mathscr{D}$, the non-zero elements describe which views share this component. 
Here, the correct structure
is that row $k$ of $\mathscr{D}$ is non-zero for only columns (views) $(1,2)$ \emph{or} $(3,4)$. As a metric of structural error, we thus compute the  RMSE distance between the MM-PCA $\mathscr{D}$-solution and the correct $\mathscr{D}$. 
For CMF the estimated $\mathscr{D}$ is directly given by the $\ell_2$ norm of the columns of its loading matrix solution (see supplement). We also evaluate the reconstruction error with regard to the signal term of each matrix.

Results, presented in Figure \ref{fig:ex1}, are based on 100 simulation runs.
MM-PCA outperforms CMF in finding the correct joint structure (lower panel in Figure \ref{fig:ex1}). 
CMF results do not provide explicit zeroes which  
makes it difficult for the user to select a threshold to decide if a signal is shared between matrices. Here, a range of threshold values were applied and the best results are reported. By contrast, when SNR is moderate to high, MM-PCA often finds the correct structure exactly.

In terms of reconstruction error (upper panel), CMF performs better than MM-PCA due to MM-PCA adhering to the more constrained (and correct) structure and due to penalization bias. Including a de\-biasing step in MM-PCA has been left for future work.

\begin{figure}
  \centering
    \includegraphics[scale=0.5]{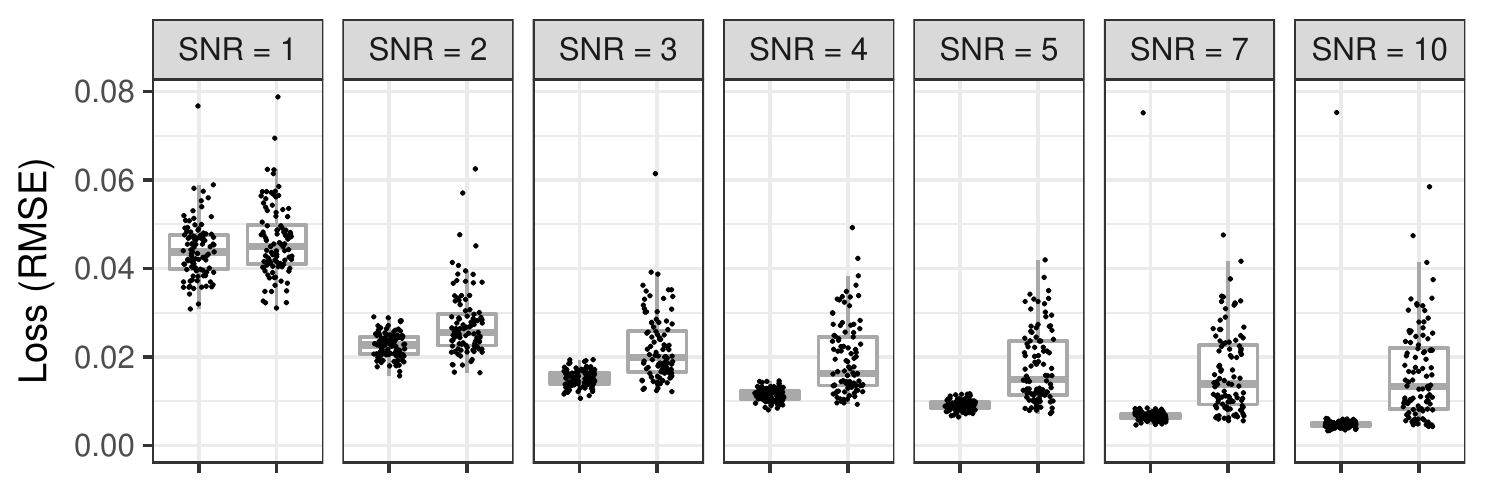}\\
    \vspace{-1.65mm}
    \includegraphics[scale=0.5]{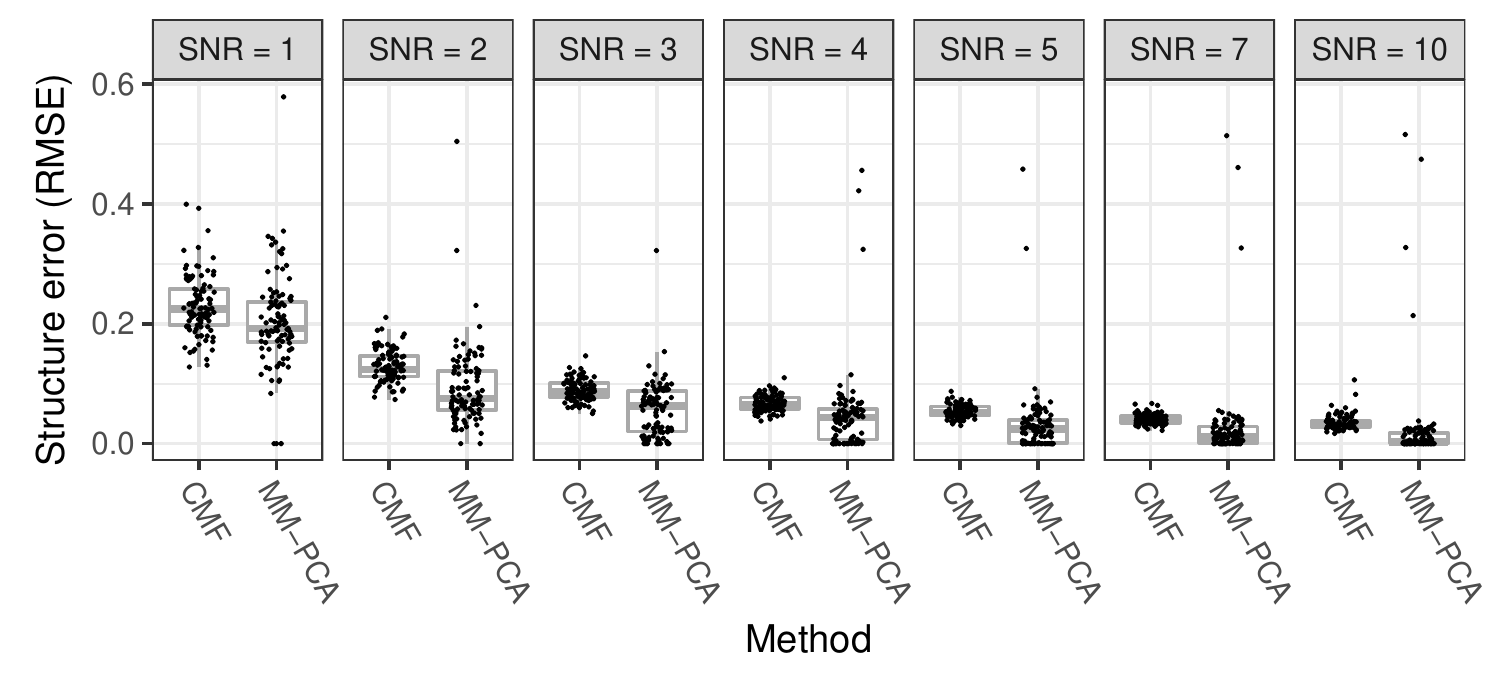}
  \caption{Simulation 1 compares MM-PCA and CMF in terms of estimating the joint structure of four matrices. The upper panel shows reconstruction loss and the lower panel shows error in joint structure. MM-PCA outperforms CMF in terms of finding the correct joint structure but has a higher reconstruction loss due to adhering to the more constrained (and correct) joint structure and due to penalization bias. RMSE, root mean square error. SNR, signal-to-noise ratio.}
  \label{fig:ex1}
\end{figure}

\subsection{Simulation 2} 
We generate three data matrices of size $n\times p$, representing three cohorts (views 1, 2 and 3) observed over a common set of features (view 4). Each observation is associated with either a loading that is common to all three matrices (i.e. globally joint) or with a loading that is unique for the matrix of that observation. We 
vary the proportion of observations that are associated with the globally joint loading.
Thus, in each matrix, each row lies either along the globally joint direction or along the individual direction of the matrix. The individual directions $v_1, v_2, v_3$ are not orthogonal, they have pairwise angle of $\pi/3$. Data is generated according to:
\begin{align}
\begin{split}
\label{eq:sim2}
X_i&=[v_iu_i^T\quad v_4w_i^T]^T + \varepsilon_i,~i=1,\ldots,3\\
u_i^T&\sim N(0, I_{[p (1-p_J)]}),~w_i^T\sim N(0, I_{[p p_J]})\\
v_4&=\tilde{v}_4/|\tilde{v}_4|_2,~\tilde{v}_4^T\sim N(0, I_{n})
\end{split}
\end{align}
where $p_J$ is the proportion of data coming from the joint loading and where each element of each $\varepsilon_i$ is independent normal with mean zero and standard deviation $0.05$.
We generate data in two settings; $(n,p)=(100,25)$ a low-dimensional case, and $(n,p)=(10,40)$ a high-dimensional case (with respect to the direction of integration).

MM-PCA is run with maximum rank 10 (the correct rank for this problem is four). CMF is run with both rank four and rank 10. Since results were slightly, but significantly, better  for CMF with rank 10 these are reported. For CMF we find the approximate rank of the solution by thresholding $\mathscr{D}$. As in Simulation 1, a range of threshold values were tried and the best  results are reported.  JIVE is run with its permutation testing method for rank selection. JIVE has an option to assume that the individual loadings are orthogonal. This assumption makes JIVE more robust to noise, though it is incorrect for this data. JIVE is run in both settings for comparison.

We evaluate the performance of the methods based on the estimated joint direction of the data. If a method selects a model that does not estimate the existence of any globally joint direction, the dominating direction of the residual is used as estimate of the joint direction. This is what a sensible user would do if the goal is to find the globally joint direction rather than to decide if such a signal exists. We choose to focus on this measure in order to facilitate a comparison of methods of different complexity.
Accuracy is defined as the estimated joint direction having a smaller angular distance to the correct joint direction than to each individual direction and to the dominating direction of the noise. 

Results from 100 simulation runs are shown in Figure \ref{fig:ex2}, the low dimensional case in Figures \ref{fig:ex2}(A,C,E) and the high dimensional case in Figures \ref{fig:ex2}(B,D,F). 
In the low dimensional case, CMF finds the globally joint loading more often than MM-PCA (Figure \ref{fig:ex2}A). In the high dimensional case, MM-PCA is better (Figure \ref{fig:ex2}B). JIVE without orthogonal constraint is the best performing method for low joint proportion, but the method performs poorly and unpredictably for higher joint proportions. This makes the method difficult to use in practice where the strength of the joint signal is unknown. As the individual signals weaken (due to the joint proportion increasing) JIVE tends to fail to separate the individual signals and misclassify them as joint. JIVE only looks for globally joint signal (here correct) which gives JIVE a smaller search space compared to both MM-PCA and CMF. JIVE with orthogonal constraint shows a robust performance, but the performance is poor compared to MM-PCA and CMF. MM-PCA shows the most robust performance overall.

Figures \ref{fig:ex2}C-D examine the methods' estimates of the rank of the globally joint signal. The correct rank is one except when the joint proportion is zero where the correct rank is zero. In the low dimensional case CMF performs best while MM-PCA performs best in the high dimensional case where CMF greatly overestimates.  JIVE overestimates when the joint and individual signal have near the same magnitude (joint proportion  0.5).
Figures \ref{fig:ex2}E-F examine the estimated rank of each matrix. The correct rank is two, the globally joint component and the individual component, except when the joint proportion is zero or one where the correct rank is one. 
Results show that the rank selection of MM-PCA works as intended in this setting and that the performance of MM-PCA is better than JIVE, despite MM-PCA addressing a much wider range of integrative analysis settings. CMF again performs well in the low dimensional case but overestimates in the high dimensional case. MM-PCA
performs well in both settings, albeit being conservative.

\begin{figure}
  \centering
    \hspace{13mm}
    \includegraphics[scale=0.5]{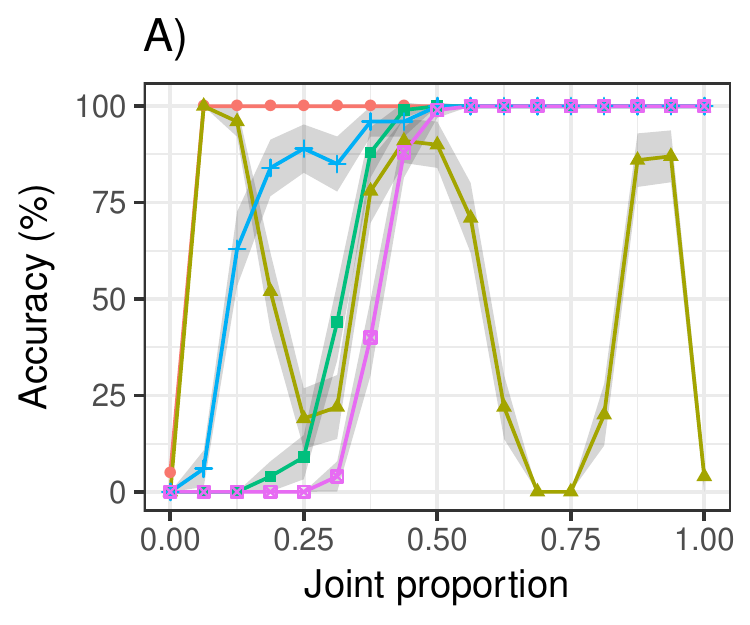}
    \hspace{-3mm}
    \includegraphics[scale=0.5]{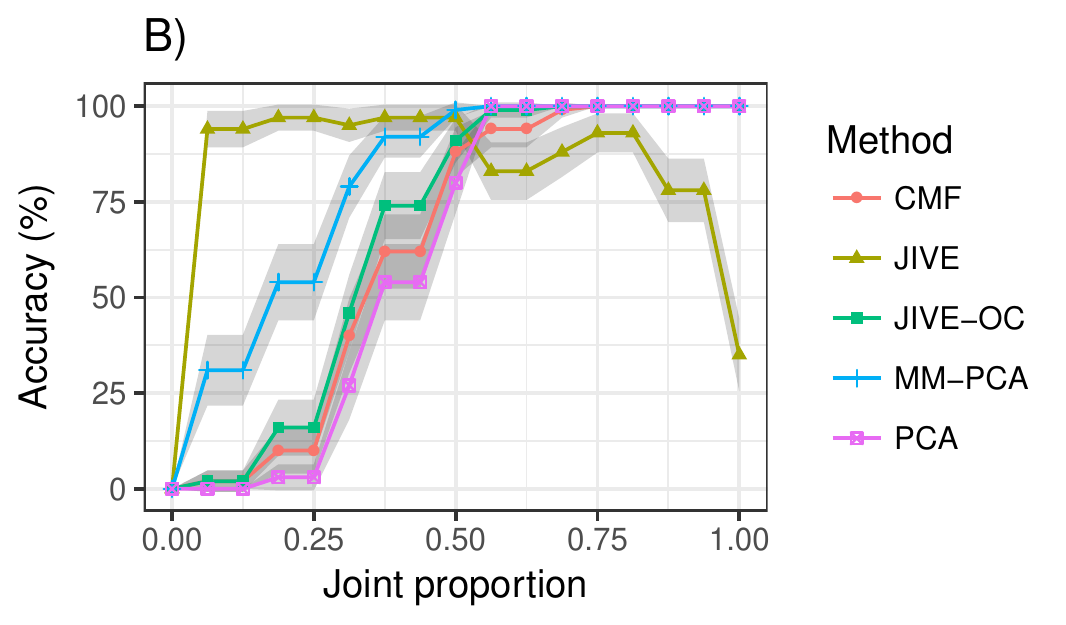}\\
    \includegraphics[scale=0.5]{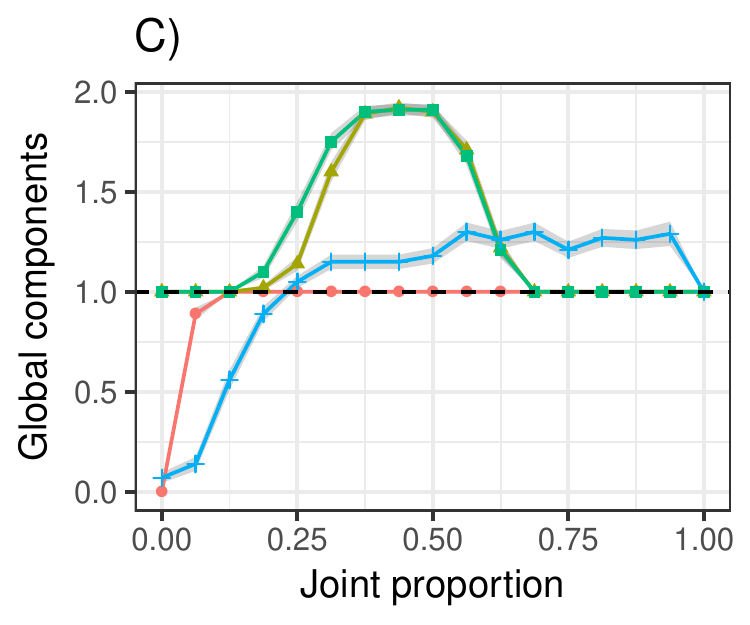}
    \includegraphics[scale=0.5]{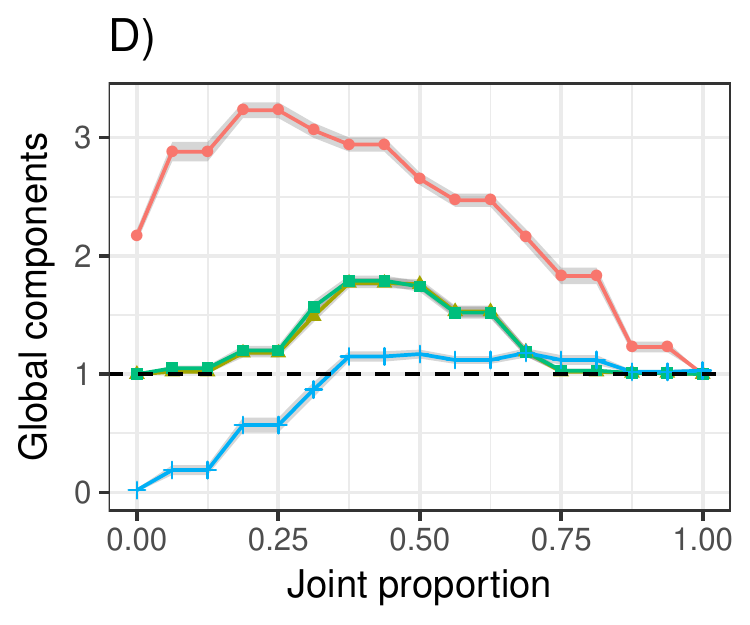}\\
    \includegraphics[scale=0.5]{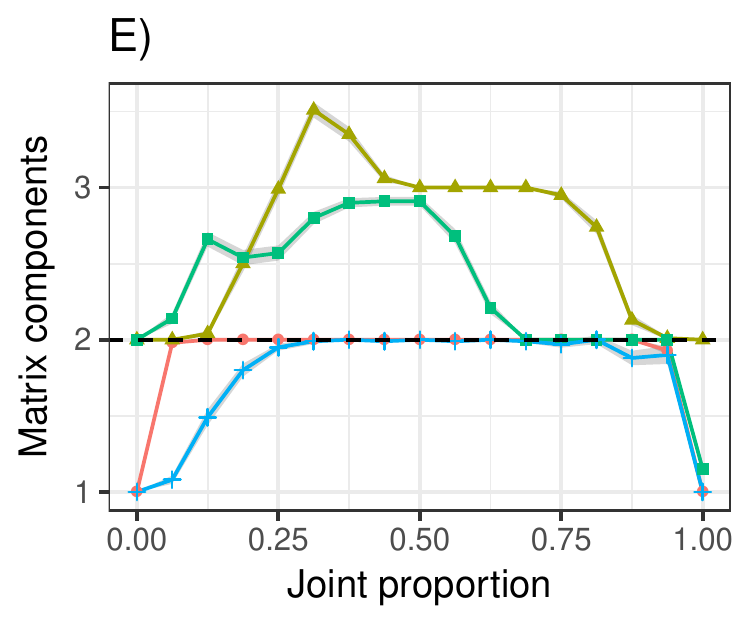}
    \includegraphics[scale=0.5]{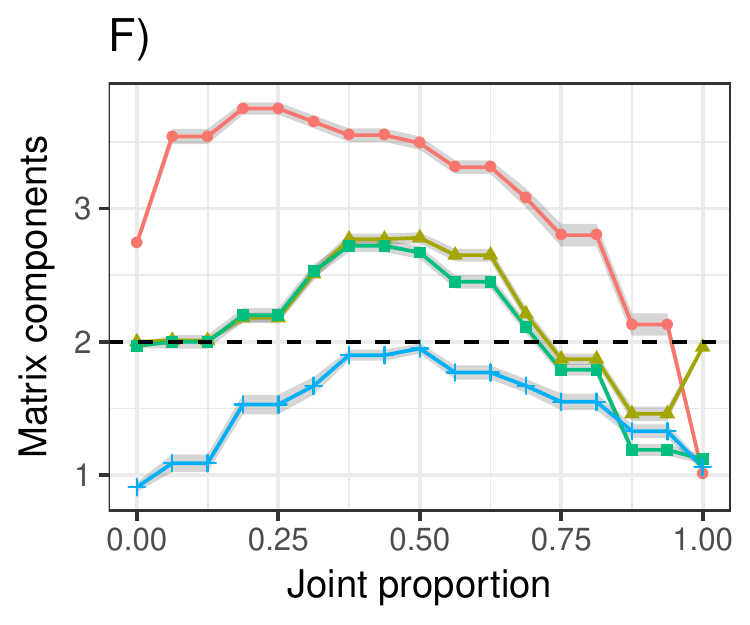}
  \caption{Comparison of MM-PCA, CMF and JIVE in a low dimensional case (A, C and E) and a high dimensional case (B, D and F). MM-PCA shows the most robust performance overall. In the low dimensional case CMF performs well, but it falls short in the high dimensional case. JIVE performs well only when the joint proportion is low. JIVE-OC (orthogonal constraint) performs robustly but poorly. A and B) Accuracy at finding the globally joint component. Ordinary PCA (performed on the concatenation of the three data matrices) is included as a baseline. C and D) Estimated number of globally joint components. E and F) Estimated rank of each matrix. 
  Dashed lines show the correct number of components.}
  \label{fig:ex2}
\end{figure}

\subsection{Simulation 3}
Simulation 3 demonstrates the ability of MM-PCA to estimate sparse loadings. We generate a matrix of size $30\times 30$. The matrix consists of two components and Gaussian noise. Each component is the product of two sparse loadings with ones at three random positions and zeros elsewhere. 

We compare the estimated components to true sparse loadings. MM-PCA is run with a maximum rank of two whereas CMF is given the correct rank (two). As in the other simulations, a range of threholds are applied to the CMF loadings to provide a comparison with MM-PCA. Accuracy is measured with Matthews correlation coefficient (MCC). 

Results from 100 simulation runs are shown in Figure \ref{fig:ex3} (upper panel). MM-PCA is able to identify the sparse loadings  with high accuracy for moderate to high SNR values. MM-PCA outperforms CMF in the entire SNR range, though this comparison is less meaningful than for Simulations 1 and 2 since CMF has not been developed to explicitly provide sparse loadings. 

We also investigate how MM-PCA rank-selection works in a sparse setting. Here, both MM-PCA and CMF, included for comparison, are run with maximum rank 3. In Figure \ref{fig:ex3} (lower panel), we see that MM-PCA performs quite well for moderate to high SNRs. In a few cases, MM-PCA results in rank zero solutions due to poor initialization and those results are excluded from the analysis.   In this setting, thresholding the CMF solutions does not result in a rank selection. This tendency for CMF to overfit was also seen in Simulation 2. 

\begin{figure}
  \centering
    \includegraphics[scale=0.5]{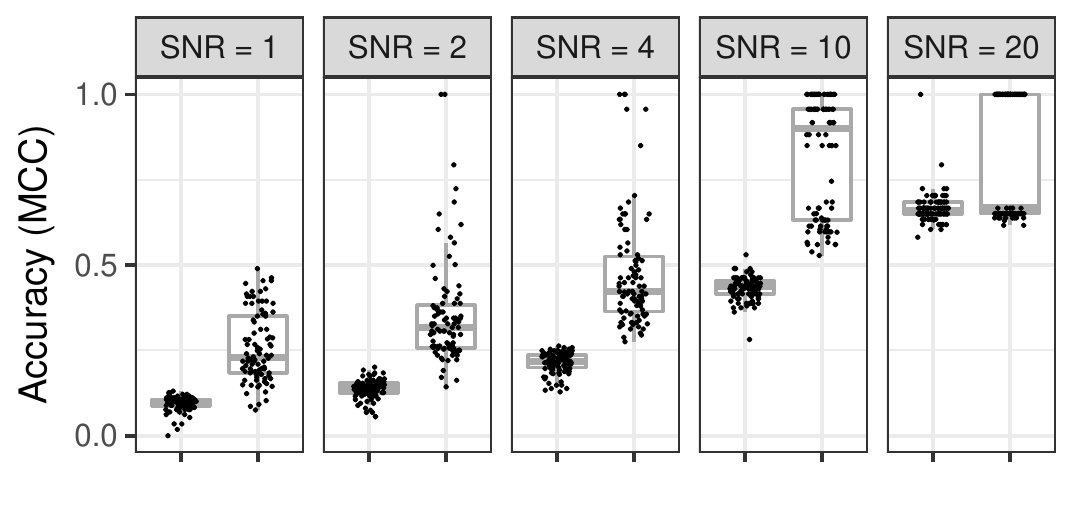}\\
    \vspace{-2.65mm}
    \hspace{-0.5mm}
    \includegraphics[scale=0.5]{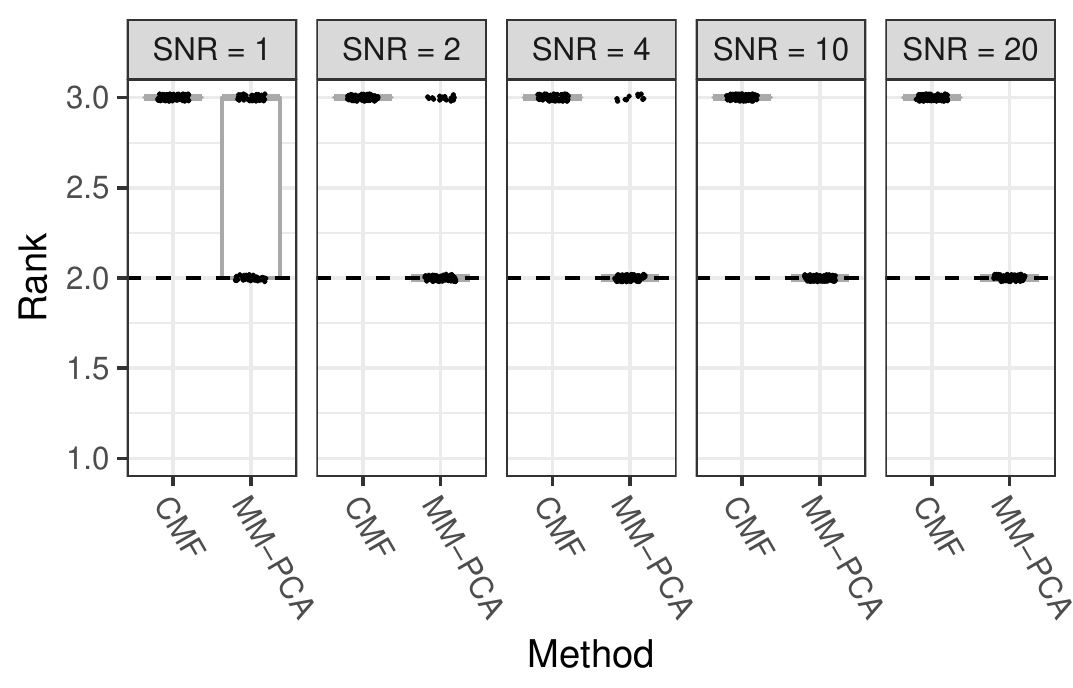}
  \caption{MM-PCA more accurately identifies the non-zero positions of the loadings and more often identifies the correct rank. MCC, Matthews correlation coefficient.  Dashed lines show the correct rank. SNR, signal-to-noise ratio.}
  \label{fig:ex3}
\end{figure}

\section{Integrative Analysis of Gene Expression and Methylation Data for Three Patient Cohorts}
\label{sec:bio}
We apply MM-PCA to a well established set of glioblastoma (GBM) data from
TCGA \citep{tcga_stuart_2013}. Earlier versions of TCGA data were used to identify four subtypes of GBM, termed proneural (PN), classical (CL), neural (NL) and mesenchymal (MS), by clustering of gene expression array profiles \citep{verhaak_integrated_2010}. Analysis of methylation data from TCGA subsequently identified a subset of patients within the PN subtype with favorable prognosis termed the G-CIMP (glioma-CpG island methylator phenotype) subgroup, more frequently observed among younger patients \citep{noushmehr_identification_2010}. Recent work has questioned the existence of the NL subtype, showing that a clustering using only tumor-intrinsic gene expression identifies three subtypes \citep{wang_tumor_2017}. 

Here, we analyze gene expression and methylation data centered around four canonical pathways of GBM; mTOR, Ras, MAPK and PI3K-Akt, focusing on processes critical for glioblastoma development and progression.

The data set consists of six data matrices: gene expressions for patient cohorts 1 and 2, methylation data for patient cohorts 2 and 3, and patient covariance matrices between cohorts 1 and 2 and between cohorts 2 and 3.
This results in five views: the three patient cohorts constitute three groups of observations, and the genes and methylation probe sites constitute two groups of variables.
Data preprocessing and normalization steps are described in detail in the supplement. After filtering on high variance genes and methylation sites, the data set to be modelled comprises three patient cohorts of sizes 96, 76 and 212, 297 genes and 305 DNA sites (Figure \ref{fig:bioex}A).

\subsection{MM-PCA Integrates Across Views and Cohorts}
\label{sec:expintegrative}

MM-PCA was applied with ten regularization levels, $\lambda_0$ between $e^{-8}$ and 1 on a logarithmic scale. The lowest cross-validation error was found for $\lambda_0\approx 0.012$. We allowed the solution to have a maximal rank of 40. The selected solution has an effective rank of 24. 
Figure \ref{fig:bioex}A shows the data as modeled by the MM-PCA solution. The six data matrix estimates shown are $\hat{X}_{ij}=V(\xi_i)D_iD_jV(\xi_j)^T, (i,j)\in\{(1,4),(2,4),(2,5),(3,5),(1,2),(2,3)\}$. View indices $1,\ldots,5$ label the five views in the order: cohort 1, cohort 2, cohort 3, genes, and methylation probes. 

In Figure \ref{fig:bioex}A matrix columns and rows are ordered according to the bi-clustering  given by the MM-PCA solution. 
The hierarchical nature of this clustering is shown with cluster trees for each view. The importance of each component for each view is shown in Figure \ref{fig:R2}C. It is apparent in Figure \ref{fig:bioex}A that clusters are pronounced in all matrices associated with each view, i.e.\ the clustering is integrative. One cluster is particularly pronounced. It can be seen as a bright yellow wide rectangular area in the lower right part of the cohort 3 methylation data. We found that the associated cluster of patients is strongly associated with low age and G-CIMP status (Section \ref{sec:expclinical}). 
MM-PCA thus detects the known connection between low age at diagnosis and G-CIMP status in an unsupervised fashion. Looking at the methylation data for cohort 2 and covariance between cohorts 2 and 3, it can also be seen that this patient cluster in cohort 3 is associated with a patient cluster in cohort 2 with strongly correlated methylation. The cluster in cohort 2 is small and its presence is thus boosted through the integration with the larger cohort 3. 

The partially shared structure, captured by $\mathscr{D}$ and $R^2$ (defined in Section \ref{sec:amvpca}), is shown in Figures \ref{fig:R2}A-D. Figure \ref{fig:R2}A shows the proportion of total variation of each matrix captured by each component. Figure \ref{fig:R2}B groups components according to which views. Thus the set of components with a given color in Figure \ref{fig:R2}B constitutes a subspace for variation in a subset of data matrices. Figure \ref{fig:R2}C shows $\mathscr{D}$ of the MM-PCA solution. Figure \ref{fig:R2}D shows the proportion of variation explained by each component, resulting in a scree-plot as commonly used in ordinary PCA.
Directed $R^2$ is shown in Figures \ref{fig:bioex}B-C. Methylation data explains variation in gene expression, in agreement with their biological cause-effect relationship.

Figure \ref{fig:R2}C provides the basis for interpreting the components illustrated in Figure \ref{fig:bioex}D-F.
Components 1, 2, 5 and 11 (Figure \ref{fig:bioex}D) capture structure that is shared globally. Components 3, 7 and 8 (Figure \ref{fig:bioex}E) capture structure that is primarily shared by 
cohorts 1 and 2 through gene expression, but also models variation in methylation data for cohort 2. Components 4, 6, 9 and 10 (Figure \ref{fig:bioex}F) capture structure that is primarily shared by cohort 2 and 3 through methylation, but also additional variation in gene expression for cohort 2. The contribution of the remaining components is small, mainly capturing variation in methylation in cohort 3

As with ordinary PCA, MM-PCA can be used to plot data in lower dimensions. In Figure \ref{fig:bioint}B we plot all observations by their first two loadings. Patients of different subtype and G-CIMP status are separated already in the first two components.
Ordinary PCA would not be able to use the entire data set, but would need to base the analysis on either cohorts 1 and 2 and gene expression or cohorts 2 and 3 and methylation data. From the MM-PCA plot one can infer a plausible subtype and G-CIMP status for the observations with unknown status. 

The MM-PCA solution can also be used for integrative bi-clustering of all observations and variables (Figure \ref{fig:bioint}A). This is in contrast to the bi-clustering in Figure \ref{fig:bioex}A which clustered the items of each view. We impute the two missing data matrices $\tilde{X}_{15}$ (cohort 1 methylation data) and $\tilde{X}_{34}$ (cohort 3 gene expressions) by multiplying the corresponding $V_\cdot$ and $D_\cdot$ matrices and construct the bigger matrix
$$
\hat{X}=\begin{bmatrix}
\hat{X}_{14} & \tilde{X}_{15} \\
\hat{X}_{24} & \hat{X}_{25} \\
\tilde{X}_{34} & \hat{X}_{35}
\end{bmatrix}=\begin{bmatrix}
V(\xi_1)D_1 \\
V(\xi_2)D_2 \\
V(\xi_3)D_3
\end{bmatrix}\begin{bmatrix}
V(\xi_4)D_4 \\
V(\xi_5)D_5
\end{bmatrix}^T.
$$
Similarly to the bi-clustering described above, we reorder columns and rows of $\hat{X}$ according to signs of elements in $[V(\xi_1)^T~V(\xi_2)^T~V(\xi_3)^T]^T$ and $[V(\xi_4)^T~V(\xi_5)^T]^T$ (positive, zero or negative). Cluster hierarchies are shown as trees. We cut the trees at two heights, using two and three components, resulting in $3^2$ and $3^3$ clusters, respectively.
The low age, G-CIMP bi-cluster considered above is again present, but now includes patients from all cohorts and some genes in addition to its methylation sites (dark blue patient cluster, dark green variable cluster).
A closer look at the relation between cluster and subtype shows that CL is predominant in the red rank 2 cluster, PN in the blue and MS in the green and purple clusters (Figure \ref{fig:bioint}C). The split of the rank 2 purple cluster into rank 3 light and dark purple clusters separates PN patients by G-CIMP status. The analysis captures several clinical variables (Section \ref{sec:expclinical}), and the clustering is thus not primarily focused on separation by subtype. GBM disease subtype has been shown to have little clinical relevance \citep{johansson_large_2018}.

\begin{figure}
  \centering
    \includegraphics[width=\textwidth]{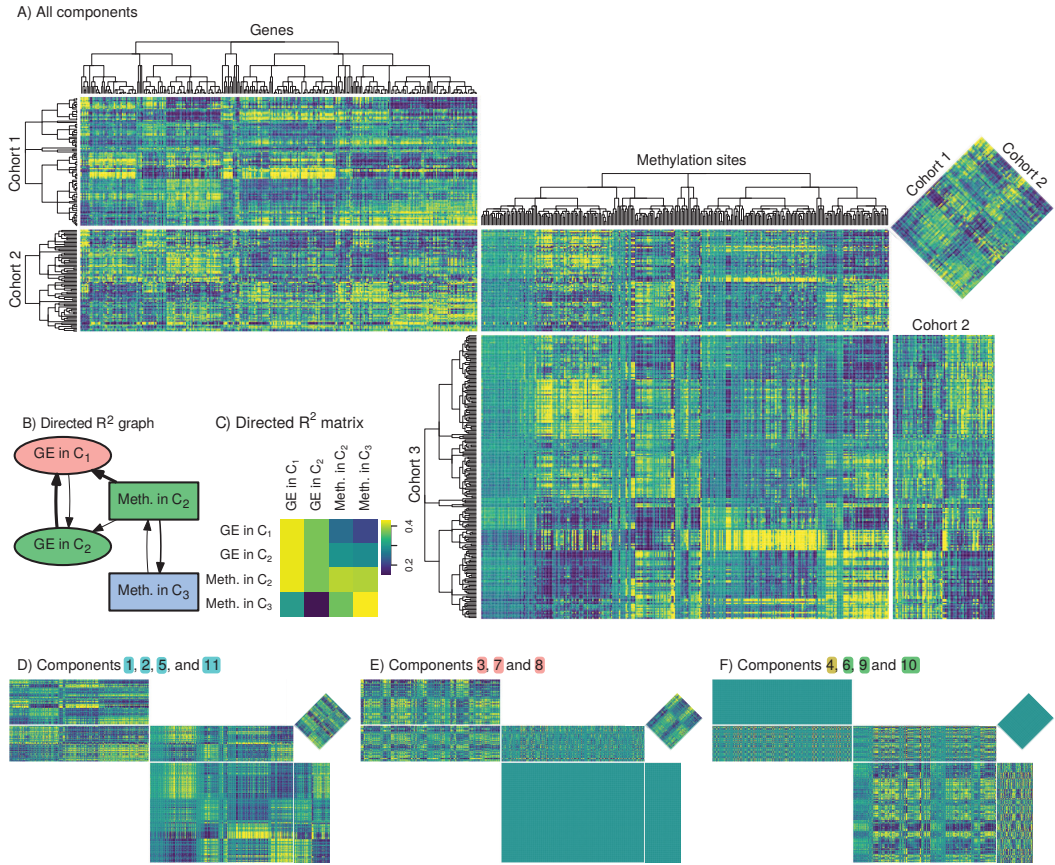}
  \caption{A) Heatmaps showing MM-PCA approximations of the glioblastoma data. Dendrograms show hierarchical integrative bi-clustering of patients and variables within each view.
  B) and C) Directed $R^2$ depicted as a graph (B) and as a matrix (C). B) A directed connection is drawn for each directed $R^2$ value above 0.35. The width of edges is proportional to the directed $R^2$ value. C) Heatmap showing all directed $R^2$ values. D) Heatmap of only globally joint components. Signal captured by these components is present in all data matrices. E) Heatmap of components shared by cohorts 1 and 2. 
 F) Heatmap of components shared by cohorts 2 and 3.
Colors in headings D-F refer to Figure \ref{fig:R2}B. Heatmap colors in A) and D-F): yellow: positive, turquoise: zero, blue: negative.}
  \label{fig:bioex}
\end{figure}

\subsection{Leading MM-PCA Components Capture Subtype, G-CIMP and Clinical Features of GBM}
\label{sec:expclinical}
To elucidate the potential biological and clinical relevance of the MM-PCA components we quantify for each: the proportion of variation in each matrix and view explained by each component (Figure \ref{fig:R2}A-C), the proportion of total variance explained by each component (Figure \ref{fig:R2}D), the relationship between each component and clinical parameters subtype, G-CIMP status, sex, age, survival time, tissue type, stemness and pathway enrichment (Figure \ref{fig:R2}E-F, Table \ref{tab:bioex}, Figures \ref{fig:compfirst}-\ref{fig:complast}). P-values are calculated using standard methods: Kruskal-Wallis test for factor variables, Fisher z-transformation for continuous variables and Cox proportional hazard model for survival times. Tissue type indicates whether a sample is from a primary tumor, recurrent tumor or from normal tissue. We use a machine learning based measure of stemness to quantify each sample's degree of oncogenic dedifferentiation \citep{malta_machine_2018}. Pathway enrichment is assessed using gene set enrichment analysis \citep[GSEA,][]{subramanian_gene_2005} using v3.0 of java\-GSEA obtained from the Broad institute. Hallmark pathway annotations where obtained from MSigDB v6.2 \citep{liberzon_molecular_2015}. 

Several components found with MM-PCA correlate significantly with clinical parameters (Figures \ref{fig:compfirst}-\ref{fig:complast}).
Component 1 separates patients with MS subtype from patients with PN subtype.
The component is also associated with genes defining the epithelial-mesenchymal transition, a hallmark pathway consistent with the gradient between the PN and MS subtype. 
Component 2 strongly separates G-CIMP patients from non G-CIMP patients as well as patients with CL subtype from other patients. 
Component 3 captures variation in the expression data that associates with several hallmark pathways including oxidative phosphorylation and mitotic spindle. The component is associated with tissue type and correlates with the stemness signature.
Component 4 captures variation in the methylation data and is associated with subtype and G-CIMP status.
Interestingly, the component associates some CL subtype patients with G-CIMP patients, suggesting a set of genes and probes where these CL patients are similar to G-CIMP patients.
Component 5 captures variation in both data types and describes variation between normal and tumorous samples. Of note, a number of NL samples have high loadings for this component, possibly indicating contamination of normal cells, which has been suggested as the basis for the NL subtype \citep{wang_tumor_2017}.
Component 6 captures patient sex, primarily by a small number of methylation probes on chromosome X (e.g.\ IRAK1 and FLNA).

\begin{figure}
  \centering
    \includegraphics{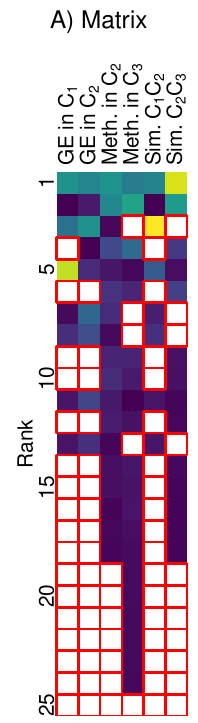}
    \includegraphics{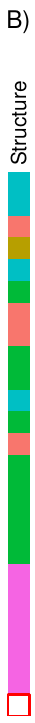}
    \includegraphics{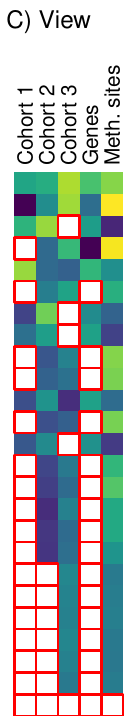}
    \hspace{-2mm}
    \includegraphics{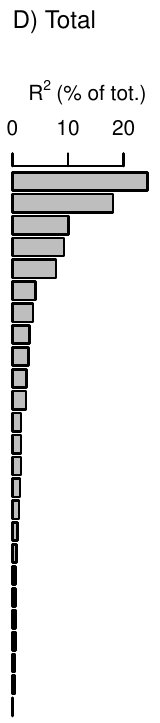}\\
    \includegraphics{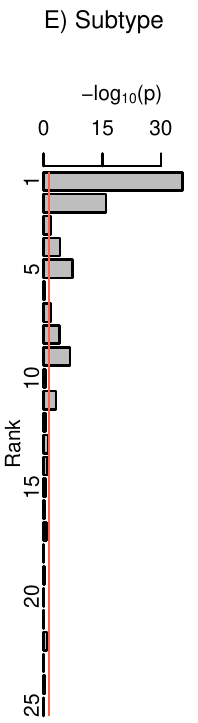}
    \hspace{-2mm}
    \includegraphics{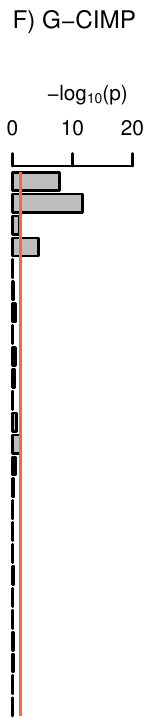}
    \hspace{-3mm}
    \includegraphics{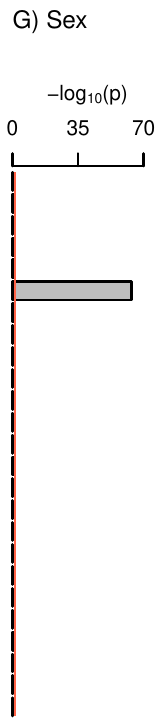}
    \hspace{-2mm}
    \includegraphics{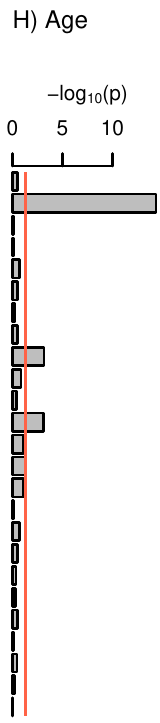}
    \includegraphics{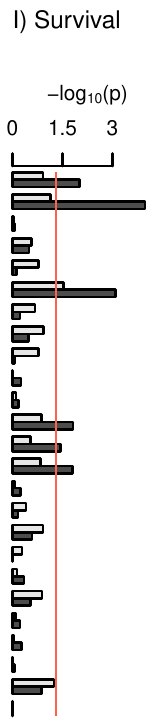}
  \caption{Joint structure estimated by MM-PCA for the glioblastoma data. A) Matrix-wise $R^2$ values capturing the proportion of variation in each matrix explained by each component. High values in bright yellow, low values in dark blue. White boxes with red borders mark exact zeros. B) Color coding of each component by the combination of views (or, equivalently, matrices) that it is active in. Green components, for example, make up a subspace that capture variation related to methylation data. C) $\mathscr{D}$, the augmented D matrix, that is column $i$ of $\mathscr{D}$ is equal to the diagonal of $D_i$. Colors as in A. D) Relative importance of each component for the approximation of the entire data set. E-I) P-values from tests for association between component loadings and clinical parameters. Red lines show significance level 0.05. I) Dark bars show association with survival. Bright bars show association for survival after correction for patient age.}
  \label{fig:R2}
\end{figure}

\begin{table}
\centering
\caption{Significance tests for association of MM-PCA component loadings with clinical parameters, stemness and gene signatures of hallmark pathways relevant to glioblastoma. Legend: $*\!*\!*$: $p<0.001$, $**$: $p<0.01$, $*$: $p<0.05$, $\cdot$ : $p<0.1$., Survival is not significantly associated, at the 0.05 level, with any component after correction for age or sex.}
\begin{tabular}{rrcccccc}
& & \multicolumn{6}{c}{MM-PCA component} \\
& & 1 & 2 & 3 & 4 & 5 & 6 \\
\cline{2-8}
Clinical: & Subtype & $*\!\!*\!\!*$ & $*\!\!*\!\!*$ & $*$ & $*\!\!*\!\!*$ & $*\!\!*\!\!*$ &  \\
& G-CIMP & $*\!\!*\!\!*$ & $*\!\!*\!\!*$ & $\cdot$ & $*\!\!*\!\!*$ &  &  \\
& Sex &  &  &  &  & $\cdot$ & $*\!\!*\!\!*$ \\
& Age &  & $*\!\!*\!\!*$ &  &  &  &  \\
& Survival & $**$ & $*\!\!*\!\!*$ &  &  &  & $*\!\!*\!\!*$ \\
& Cohort &  & $**$ & $**$ &  & $*\!\!*\!\!*$ & $**$ \\
& Tissue type & $\cdot$ &  & $*\!\!*\!\!*$ &  & $*\!\!*\!\!*$ &  \\
\cline{2-8}
Stemness: & mRNA & $*\!\!*\!\!*$ &  & $*$ &  & $**$ &  \\
& Methylation & $*\!\!*\!\!*$ &  &  &  &  &  \\
\cline{2-8}
Hallmark pathways: & Epith. mesench. trans. & $**$ &  &  &  &  &  \\
& Oxidative phosph. &  &  & $*$ &  &  &  \\
& Mitotic spindle &  &  & $**$ &  &  &  \\
& KRAS signaling DN & $*$ &  &  &  &  &  \\
\end{tabular}
\label{tab:bioex}
\end{table}

\begin{figure}
  \centering
    \includegraphics[width=\textwidth]{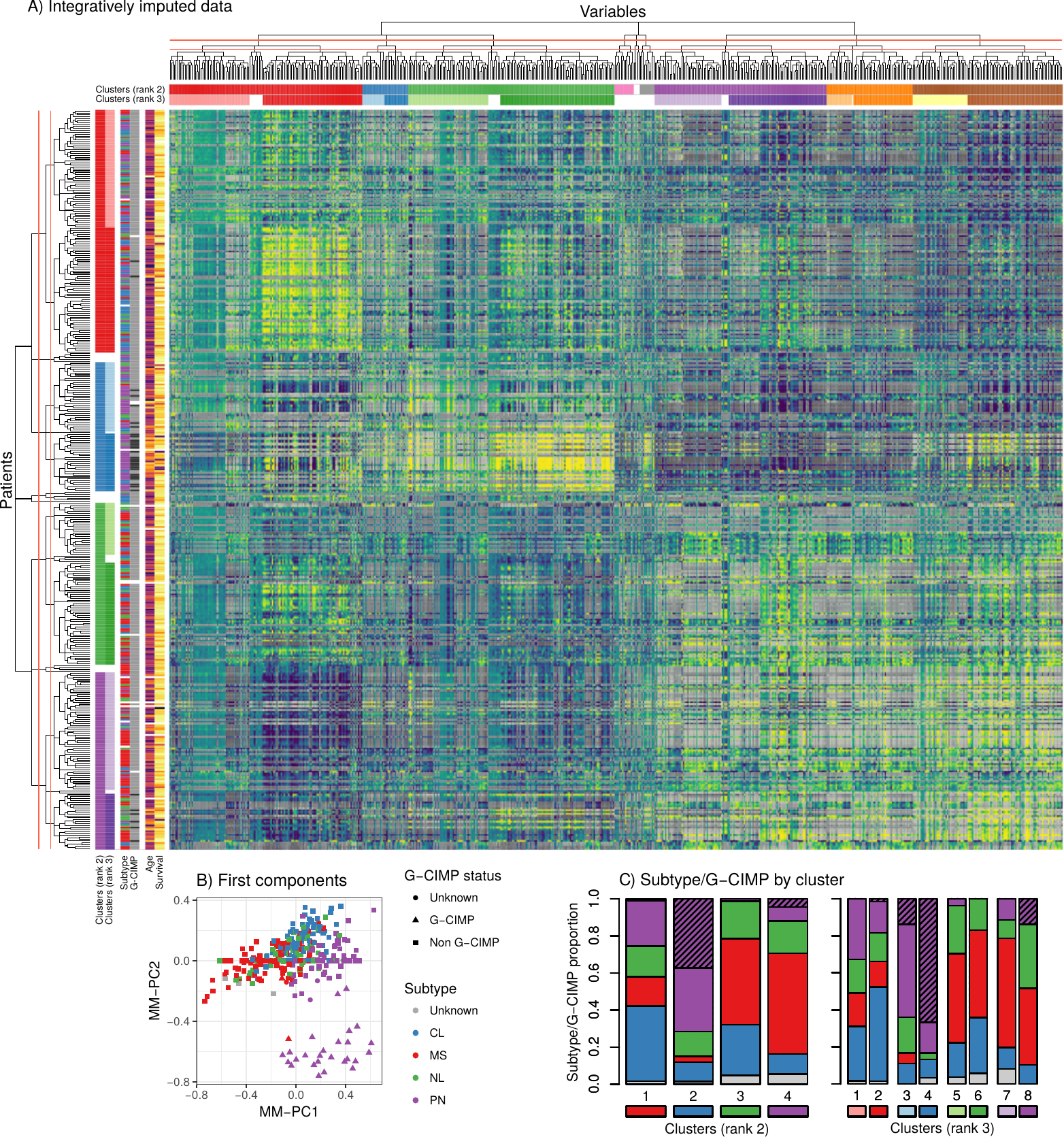}
  \caption{Integrative imputation and bi-clustering. A) Heatmap of the MM-PCA approximation of the glioblastoma data and imputation of gene expression and methylation for the cohorts that miss each respective data type. Heatmap colors: yellow: positive, turquoise: zero, blue: negative. Imputed data is gray, light: positive, dark: negative. Dendrograms show hierarchical clustering of patients and variables. Red lines across dendrograms show tree cuts into clusters. Color bars show cluster assignment and clinical parameters of patients. B) MM-PCA plot of all patients by values in first two components. C) Proportion of disease subtype in each patient cluster. Colors in bars map to subtypes in B). Colors below bars map to patient clusters in A).}
  \label{fig:bioint}
\end{figure}

\section{Discussion}
\label{sec:disc}
We present a method for multi-group and multi-view data integration, MM-PCA, capable of identifying partially shared components between any subset of data matrices. The goal of MM-PCA is to provide a general data integration pipeline that is applicable to heterogeneous data where global integration would be too restrictive. 

MM-PCA provides a sparse, low-rank representation of multiple data matrices. The MM-PCA solution is easily interpreted since; (i) each component is identified as associated with a subset of data matrices; (ii) sparse component loadings directly translates the MM-PCA solution to an integrative bi-clustering procedure; and (iii) MM-PCA automatically selects the rank of the approximation as part of the procedure. 

We show through simulations that MM-PCA performs well in a range of settings including low- and high-dimensional tasks, small and large proportion of joint variation, and low to high signal-to-noise ratio. MM-PCA exhibits a robust and stable performance across this range, whereas competitors CMF and JIVE perform poorly in some of these settings. This provides support for using MM-PCA in real data settings where the underlying structure is unknown. We apply MM-PCA to an 'omic data set of glioblastoma to illustrate the power of the method as an explorative tool for complex data integration. 

MM-PCA constitutes a complex optimization task. The current implementation takes 20 minutes to analyze the data treated in Section \ref{sec:bio} on a 4 core 3.5 GHz desktop computer. On simulated data, runtimes are on par with CMF (time per iteration is somewhat lower for MM-PCA) despite MM-PCA adressing a more complex problem.  Future work could be directed at providing approximate solutions or optimization strategies to improve runtimes.

To further extend the applicability of MM-PCA, the method should be generalized to other loss functions (e.g.\ for binary data) and tensor data. This is left for future work. 

\section{Software}
\label{software}
MM-PCA is available as an R package on CRAN.

\section{Supplementary Material}
\label{supplementary}

Supplementary material is available in the appendix.

\section*{Acknowledgments}
This work was supported by grants from Vetenskapsrådet (VR 2013-05101) and the Swedish Foundation for Strategic Research (SSF BD15-0088).

{\it Conflict of Interest}: None declared.

\bibliographystyle{biorefs}

\bibliography{report}

\markboth%
{J. Kallus and others}
{MM-PCA}

\appendix
\section{Literature Review}
This literature review is summarized in Table~\ref{tab:review} and the different data integration problems are illustrated in Figure~1. The simplest case is the integration of data from multiple cohorts observed across a common set of features, commonly referred to as multi-group integration (Figure~1A). An early extension of PCA for multi-group data is CPCA \citep{flury_common_1984}. It summarizes a data set consisting of several cohorts with one matrix of loadings that is common for all cohorts. The multi-view (multi-source, multi-modal) data integration problem focuses instead on one cohort where multiple sets of features or measurements are available (Figure~1B). The popular iCluster algorithm \citep{shen_integrative_2009} is an example of multi-view integration for genomic data. Below we review methods for multi-group or multi-view data integration together since they are algorithmically similar.

Several classical methods for multi-view (or multi-group) integrative data analysis are based on PCA or SVD. Canonical correlation analysis (CCA) \citep{hotelling_relations_1936} and partial least squares (PLS) \citep{wold_causal_1974} are perhaps the most used and studied. They can only be applied to data sets comprising two matrices, and they focus only on finding joint structure or parameters to explain both matrices together. Hierarchical PCA, hierarchical PLS \citep{wold_hierarchical_1996}, generalized SVD \citep{van_loan_generalizing_1976}  and higher order generalized SVD \citep{ponnapalli_higher-order_2011} are extensions to more than two data matrices but still limited to finding a joint structure shared by all matrices in a dataset. 

JIVE \citep{lock_joint_2013}, DISCO-SCA \citep{schouteden_sca_2013} and O2-PLS \citep{trygg_o2-pls_2002} are methods that models data as a sum of both globally joint signal and signal that is individual to each matrix. They estimate model parameters by iteratively alternating between fitting a global model for the joint structure and fitting individual models to the residuals in each matrix. O2-PLS applies only to datasets of two matrices but OnPLS \citep{lofstedt_onpls-novel_2011} generalizes it to an arbitrary matrix count. AJIVE \citep{feng_angle-based_2018} uses the same model as JIVE but bases rank estimation on random matrix theory instead of bootstrapping. GIPCA \citep{zhu_generalized_2018} generalizes JIVE to allow for different data types and also allows for the presence of missing values. 

The method sMVMF \citep{wang_sparse_2015} provides a decomposition of multi-group data into global and individual signal while imposing sparsity penalties on the global and individual loadings.

Simultaneous data integration for multi-group \emph{and} multi-view data (Figure~1C) was first addressed by bi-modal OnPLS  \citep{lofstedt_bi-modal_2012}, an extension of OnPLS. The method can decompose such datasets into global and individual signal. Linked matrix factorization \citep{oconnell_linked_2017} later extended JIVE to the same generality, with the additional ability to handle elementwise missing data. BIDIFAC \citep{park_integrative_2019} allow signal to be group-global and view global, in addition to the global and individual signal. Integrative analysis of data that is both multi-view and multi-group will become increasingly relevant in molecular biology \citep{richardson_statistical_2016}.

MOFA \citep{argelaguet_multiomics_2018} is the only method for multi-view data that can find joint signal shared only for a \emph{subset} of matrices. Like sMVMF, MOFA also performs feature selection. It does so by using a sparsity inducing prior.

CMF \citep{klami_group-sparse_2014} builds on the augmented multi-view data (AMD, Figure~1D) framework to perform multi-group and multi-view integration with loadings that are shared by a subset of matrices. The partial sharing problem for multi-view integration was, as mentioned above, addressed in MOFA. Both of these methods utilize a Bayesian setup. CMF uses the model $X_{ij}=U_iU_j^T+\varepsilon_{ij}$ where each $U$ and $\varepsilon$ follow a Gaussian distribution. The model parameters are estimated using variational Bayesian inference.

\begin{sidewaystable}
\centering
\caption{Summary of reviewed methods for integrative analysis of data in matrix form with a focus on recent methods. Methods are compared in terms of the types of datasets they can handle, how their analysis decomposes the data and if interpretation of the analysis is facilitated by estimating sparse loadings. In the column Signal decomposition, "Any subset" means that joint signal for any subset of matrices can be found. This includes global and individual signals. *CMF and MOFA do not produce a hard decision for the subset signal decomposition, thus reducing their interpretable value.}
\begin{tabular}{cccccc}
  \multirow{2}{*}{Method}
    & \multicolumn{3}{c}{Data support}  & \multirow{2}{*}{\thead{Signal \\ decomposition}}
    & \multirow{2}{*}{\thead{Sparse \\ loadings}} \\ \cline{2-4}
  & Setup & \thead{Matrix \\ count} & \thead{Missing \\ data} \\  \hline
     CCA,       Hotelling, 1936  &       Multi-view         &     2       &   No   &          Global         &    No      \\
     PLS,        H. Wold, 1974   &       Multi-view         &     2       &   No   &          Global         &    No      \\
    GSVD,        Van Loan, 1976  &       Multi-view         &     2       &   No   &          Global         &    No      \\
    CPCA,         Flury, 1984    &      Multi-group        & Arbit.   &   No   &          Global         &    No      \\
   HOGSVD,      Ponnapalli, 2011 &       Multi-view         & Arbit.   &   No   &          Global         &    No      \\
  Hier. PCA,     S. Wold, 1996   &       Multi-view         & Arbit.   &   Yes  &          Global         &    No      \\
  Hier. PLS,     S. Wold, 1996   &       Multi-view         & Arbit.   &   Yes  &          Global         &    No      \\
   O2-PLS,        Trygg, 2002    &       Multi-view         &     2       &   No   &   Glob. \& ind.  &    No      \\
  DISCO-SCA,    Schouteden, 2013 & Multi-group or -view  &     2       &   No   &   Glob. \& ind.  &    No      \\
    OnPLS,       L\"ofstedt, 2011  &       Multi-view         & Arbit.   &   No   &   Glob. \& ind.  &    No      \\
    JIVE,          Lock, 2013    &       Multi-view         & Arbit.   &   No   &   Glob. \& ind.  &    No      \\
    AJIVE,         Feng, 2018    &       Multi-view         & Arbit.   &   No   &   Glob. \& ind.  &    No      \\
    sMVMF,         Wang, 2015    &      Multi-group        & Arbit.   &   No   &   Glob. \& ind.  &   Yes      \\
    GIPCA,         Zhu, 2018     &       Multi-view         & Arbit.   &   Yes  &   Glob. \& ind.  &    No      \\
    MOFA,       Argelaguet, 2018 &       Multi-view         & Arbit.   &   Yes  &       Any subset*       &   Yes      \\
     LMF,        O'Connel, 2017  & Multi-group and -view  &     3       &   Yes  &   Glob. \& ind.  &    No      \\
Bi-mod. OnPLS,   L\"ofstedt, 2012  & Multi-group and -view  & Arbit.   &   No   &   Glob. \& ind.  &    No      \\
BIDIFAC,   Park, 2019  & Multi-group and -view  & Arbit.   &   Yes   &   Glob. \& ind.  &    No      \\
     CMF,         Klami, 2014    &  Augmented multi-view   & Arbit.   &   Yes  &       Any subset*       &    No      \\
   MM-PCA,        This article   &  Augmented multi-view   & Arbit.   &   Yes  &        Any subset       &   Yes
\end{tabular}
\label{tab:review}
\end{sidewaystable}

\section{Generalized Euler Parametrization}
\label{sec:geneuler2}

We present the following corollary to the results of \citep{hoffman_generalization_1972}. It states that $k$-frames can be factorized into $pk-k(k+1)/2$ Givens rotation matrices.

\begin{corollary}
\label{geneuler}
For $p\geq k$, an arbitrary $p$-dimensional $k$-frame $V$ can be expressed as $R_1R_2\cdots R_mI_{pk}$, where $m=pk-k(k+1)/2$, $R_i$ are generalized Euler rotations and $I_{pk}$ are the first $k$ columns of $I_p$.
\end{corollary}
\begin{proof}
Let $V$ be an arbitrary $k$-frame and let $W$ be any rotation such that $V=WI_{pk}$. Such $W$ exist by the definition of $k$-frames. \cite{hoffman_generalization_1972} gives $W=\prod_{1\leq i<j\leq p}R_{ij}$ which can be factorized into $(\prod_{(i,j)\in \mathcal{A}}R_{ij})\allowbreak(\prod_{(i,j)\in \mathcal{A}^C}R_{ij})$ where $\mathcal{A}=\{(i,j):1\leq i<j\leq p\land i\leq k\}$ and $\mathcal{A}^C=\{(i,j):k<i<j\leq p\}$ since the factors of $W$ are ordered so that this can be done without reordering them. Furthermore $(\prod_{(i,j)\in \mathcal{A}^C}R_{ij})I_{pk}=I_{pk}$ since all indices in $\mathcal{A}^C$ are greater than $k$. Thus $V=(\prod_{(i,j)\in \mathcal{A}}R_{ij})I_{pk}$. Lastly $|\mathcal{A}|=m$.
\end{proof}

\subsection{Inverse}
In order to find a good initial value for the optimization problem we need to find the angles that correspond to a given $k$-frame $V^*$, i.e. $\xi=V^{-1}(V^*)$. \cite{hoffman_generalization_1972} states this inverse for rotations. To utilize the k-frames setup for MM-PCA, we here generalize the rotation inverse to $k$-frames in the following way. Let $U^{-1}(U)$ be the inverse for rotations given in \citep{hoffman_generalization_1972}, i.e.\ the function that gives the angles that correspond to a given rotation matrix $U$. Using the orthogonal complement $V^\perp$ of $V^*$ and that $[V^*\quad V^\perp]$ is a rotation matrix we set $V^{-1}$ to be the first $k$ columns of $U^{-1}([V^*\quad V^\perp])$. One step remains. Let $V^{\perp -}$ be equal to $V^{\perp}$ except for the last column, for which the signs are changed. Then $V^{\perp -}$ is also an orthogonal complement to $V^*$. Thus, let $V^{-1^-}$ be the first $k$ columns of $U^{-1}([V^*\quad V^{\perp -}])$. Either $V^*=V(V^{-1}(V^*))$ or $V^*=V(V^{-1^-}(V^*))$. The candidate for which the equality holds is used to invert $V^*$.

\section{Gradients of the Objective Function}
\label{sec:obj2}

Let $L(D,\xi_1,\ldots,\xi_n)$ be the objective function given in the article. The gradient with respect to element $(a,b), a > b$ of $\xi_i$ is
\begin{align*}
\frac{\partial L}{\partial (\xi_i)_{ab}}=-2\sum_{j:(i,j)\in\mathcal{S}}\text{tr}(A_{iab}^T(M_{ij}\odot(X_{ij}-V(\xi_i)D_iD_jV(\xi_j)^T))V(\xi_j)D_iD_j\cdot\\
\cdot B_{iab}^T(\frac{\partial R_{iab}}{\partial (\xi_i)_{ab}})^T)-2\sum_{j:(j,i)\in\mathcal{S}}\text{tr}(A_{iab}^T(M_{ij}^T\odot(X_{ij}^T-V(\xi_i)D_iD_jV(\xi_j)^T))V(\xi_j)\cdot\\
\cdot D_iD_jB_{iab}^T(\frac{\partial R_{iab}}{\partial (\xi_i)_{ab}})^T)+\lambda_3n_v^{-1}\text{tr}(A_{iab}^T\text{sign}(V(\xi_i)D_i)D_iB_{iab}^T(\frac{\partial R_{iab}}{\partial (\xi_i)_{ab}})^T)+\\
+\lambda_4n_v^{-1}\sum_{j=1}^{p_i}\text{tr}(||L_{ij}V(\xi_i)D_i||_2^{-1}A_{iab}^TL_{ij}^TL_{ij}V(\xi_i)D_iD_iB_{iab}^T(\frac{\partial R_{iab}}{\partial (\xi_i)_{ab}})^T)
\end{align*}
where $A_{iab},R_{iab}((\xi_i)_{ab})$ and $B_{iab}$ are defined by
\begin{align*}
A_{iab}R_{iab}((\xi_i)_{ab})B_{iab}=R_1R_2\cdots R_mI_{p_ik}=V(\xi_i)
\end{align*}
(see corollary \ref{geneuler}), $L_{ij}$ is a row vector with a one at position $j$ and zeros elsewhere and
\begin{align*}
\frac{\partial R_{iab}}{\partial (\xi_i)_{ab}}=\begin{bmatrix}0_{a-1} & 0 & 0 & 0 & 0\\
0 & -\sin((\xi_i)_{ab}) & 0 & -\cos((\xi_i)_{ab}) & 0\\
0 & 0 & 0_{b-a-1} & 0 & 0\\
0 & \cos((\xi_i)_{ab}) & 0 & -\sin((\xi_i)_{ab}) & 0\\
0 & 0 & 0 & 0 & 0_{p_i-b}\end{bmatrix}.
\end{align*}
If $a\leq b$ then $\partial L/\partial(\xi_i)_{ab}=0$. The gradient with respect to $D_{\cdot i}$ is
\begin{align*}
\frac{\partial L}{\partial D_i}=I(-2\sum_{j:(i,j)\in\mathcal{S}}V(\xi_i)^T(M_{ij}\odot(X_{ij}-V(\xi_i)D_iD_jV(\xi_j)^T))V(\xi_j)D_j-\\
-2\sum_{j:(j,i)\in\mathcal{S}}V(\xi_i)^T(M_{ij}^T\odot(X_{ij}^T-V(\xi_i)D_iD_jV(\xi_j)^T))V(\xi_j)D_j)+\\
+\lambda_1\text{sign}(D_{\cdot i})+\lambda_2 D_{\cdot i}\oslash\sqrt{(\mathscr{D}\odot\mathscr{D})\mathbf{1}}+\lambda_3n_v^{-1}V(\xi_i)^T\text{sign}(V(\xi_i)D_i)+\\
+\lambda_4n_v^{-1}\sum_{j=1}^{p_i}||(V(\xi_i)D_i)_{j\cdot}||_2^{-1}(V(\xi_i)_{j\cdot}\odot(V(\xi_i)D_i)_{j\cdot})
\end{align*}
where $\oslash$ is elementwise division and $\mathbf{1}$ is a column vector with each element being 1.

\section{Optimization, Algorithmic Details}
\label{sec:opt2}

The most computationally demanding task in MM-PCA is to compute the gradients of the objective with regard to the matrices $\xi_i$. This task has been parallelized across views and uses $n_v$ cores simultaneously.

To improve the convergence rate, the $\ell_1$ penalties are approximated by a smooth function based on tangens hyperbolicus as in \citep{trendafilov_projected_2006}.

The matrices $\xi_i$ contain angles in radians, so they are on the scale $[-\pi, \pi]$. When selecting tolerance and stopping criterion for BFGS it is convenient if all variables are on the same scale. Therefore the data is rescaled so that the biggest singular value, among all the singular values of all data matrices, is $\pi^2$. This results in the variables in $D$ being on approximately the same scale as $\xi_i$.

The penalty parameters $\lambda$ are rescaled to remove their dependency on the magnitude and size of the data matrices. In the objective function, each $\lambda_i$ is multiplied with a factor $c^{3/2}$, where $c=|\mathcal{S}|^{-1}\sum_{(i,j)\in\mathcal{S}}||X_{ij}||_F$, the mean Frobenius norm of the data matrices. The following derives this factor. Scaling all data by some constant $c$ should lead to scaling all matrix approximations by the same factor: $cV_iD_iD_jV_j^T=V_i(\sqrt{c}D_i)(\sqrt{c}D_j)V_j^T$. Thus, if the sum of all penalty terms before scaling the data was $\lambda^Tf(D_\cdot)$, then after scaling, the sum becomes $\lambda^Tf(\sqrt{c}D_\cdot)=\sqrt{c}\lambda^Tf(D_\cdot)$. The last equality is due to penalties being linear in the scale of $D_\cdot$. Multiplying the data and approximations by $c$ gives the objective $\sum_{(i,j)\in\mathcal{S}}||cX_{ij}-cV_iD_iD_jV_j^T||_F^2+\lambda^Tf(\sqrt{c}D_\cdot)=c^2\sum_{(i,j)\in\mathcal{S}}||X_{ij}-V_iD_iD_jV_j^T||_F^2+\sqrt{c}\lambda^Tf(D_\cdot)$ which has the same minimizing argument as $\sum_{(i,j)\in\mathcal{S}}||X_{ij}-V_iD_iD_jV_j^T||_F^2+c^{-3/2}\lambda^Tf(D_\cdot)$. Thus, setting $\lambda$ to $c^{3/2}\lambda$ removes the dependency of $\lambda$ on the scale of the data. This results in the $\lambda$ parameters being on the same scale regardless of the scale and size of the data.

\subsection{Initialization}
As initial values in optimization, we use an approximation of the solution under the assumptions that all components are globally joint. 

We use SVD to set $V_i$ to the loadings of the concatenated matrix consisting of all matrices (some transposed) with a view $i$. We find singular values $\Lambda_{ij}=V^T_iX_{ij}V_j$ for each matrix. We factorize $\Lambda_{ij}$ into $D_iD_j$ by greedily choosing signs for each diagonal element in each $D_\cdot$ and then finding the magnitude of the elements with continuous optimization. Finally, the initial values of $\xi_\cdot$ are found using the inverse $k$-frame operation, $V^{-1}$, defined above. 

\section{Integrative Analysis of Gene Expression and Methylation Data for Three Patient Cohorts}
Figures \ref{fig:compfirst}-\ref{fig:complast} show association of the first six components with clinical features. Figure \ref{fig:survclust} shows survival plots stratified by cluster assignment.

\subsection{Data Selection, Preprocessing and Normalization} \label{sec:expdata} Genomic data from GBM cancer patients from TCGA is publicly available through UCSC cancer genomics browser \citep{goldman_zhu_2014}. The data contains log-transformed RNASeq counts belonging to 20530 genes for 172 patients. It also contains measurement of methylation at 27578 DNA sites for a partially overlapping group of 288 patients. The measurements are ratios (so-called $\beta$-values) with 0 being unmethylated and 1 being methylated. We quantile normalized each RNASeq sample, transforming the counts to follow a standard normal distribution based on the order of genes by count. This removes the dependency of counts on the sequencing depth and ensures that data follows a normal distribution. For each methylation sample we subtracted the sample mean from each measurement. Next, we removed sites having median absolute deviation equal to zero, that is constant across at least half of the patients. This lowered the site count by 2597.
We removed half of the genes and half of the sites with lowest variance.
We centered genes and scaled them to have unit variance. We also centered the data for each methylation site. We divided the patients into three cohorts based on data availability. Cohort 1 contains 96 patients for which we have gene expression data but not methylation data. Cohort 2 contains 76 patients for which we have both gene expression and methylation data. Cohort 3 contains 212 patients for which we have methylation data but no gene expression data. We focus our analysis on four cellular pathways (mTOR, Ras, MAPK and PI3K-Akt) that are involved in cell proliferation, survival, migration and angiogenesis (forming of new blood vessels) and are involved in GBM pathogenesis \citep{pearson_targeting_2017}. We downloaded gene lists for each pathway from KEGG \citep{kanehisa_goto_2000}. A total of 760 genes were in at least one of the four pathways. Of these 760 genes 297 remain after the above described gene filtering. Out of the remaining DNA sites 305 sites are believed to be associated with any of these 297 genes, according to the Infinium MethylationEPIC v1.0 B4 Manifest file obtained from Illumina. We use all genes and DNA sites that remain after gene and site filtering to calculate covariance matrices between cohorts 1 and 2 as well as between cohorts 2 and 3. These two matrices are included in our analysis with the role of prior information of covariance between the data matrices of gene expressions, and the data matrices of methylation, respectively.
A small proportion (99 elements) of data are missing in the methylation data.
Lastly, we scale each matrix so that the Frobenius norm of its first SVD component is proportional to the number of observations (rows) in the matrix. To summarize, the data consists of three patient cohorts of sizes 96, 76 and 212, 297 genes and 305 DNA sites.

\begin{figure}
  \centering
    \includegraphics[width=\textwidth]{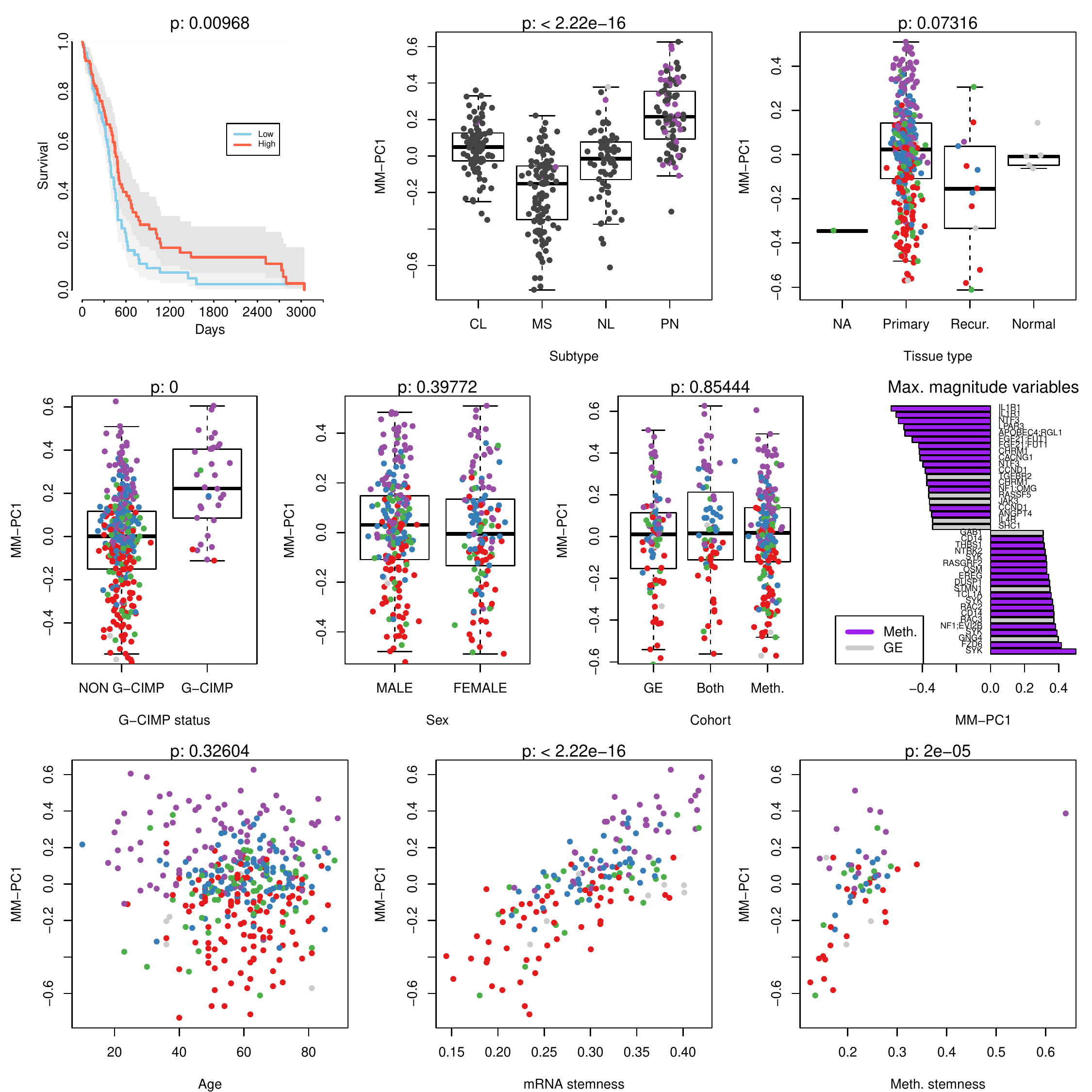}
  \caption{MM-PCA component 1. Colors in the subtype box-plot are purple for patients with G-CIMP phenotype, black for non G-CIMP and gray for unknown. Colors in other scatter- and box-plots are as in Figure \ref{fig:bioint}B: CL: blue, MS: red, NL: green, PN: purple, unknown: gray. Tests for computing p-values are specified in section 4.2 of the article. The survival plot contrasts survival times of patients with positive loadings to survival times of patients with negative loadings. The horizontal bar-plot shows the variables with maximal and minimal loadings.}
  \label{fig:compfirst}
\end{figure}
\begin{figure}
  \centering
    \includegraphics[width=\textwidth]{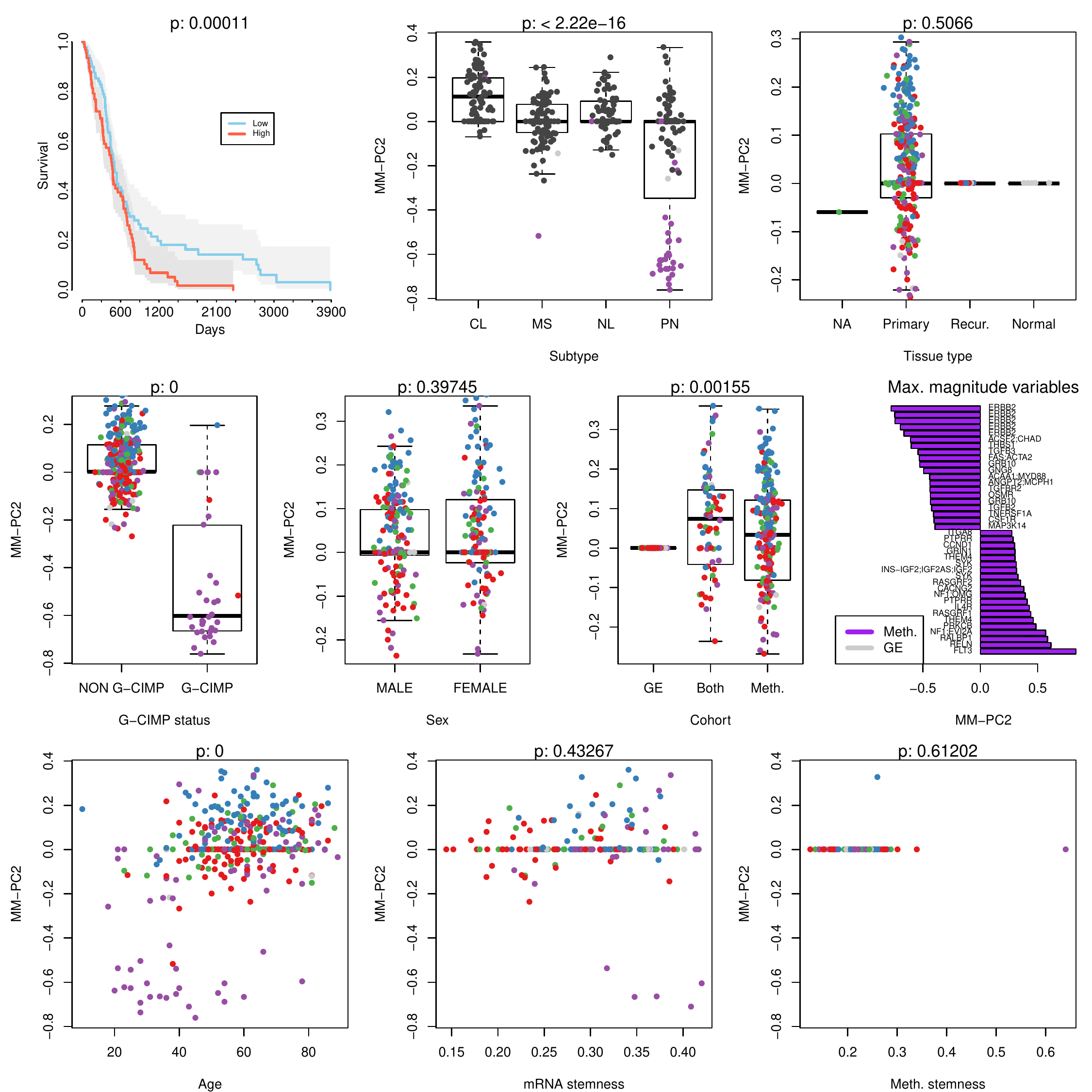}
  \caption{MM-PCA component 2. See caption of Figure \ref{fig:compfirst}.}
\end{figure}
\begin{figure}
  \centering
    \includegraphics[width=\textwidth]{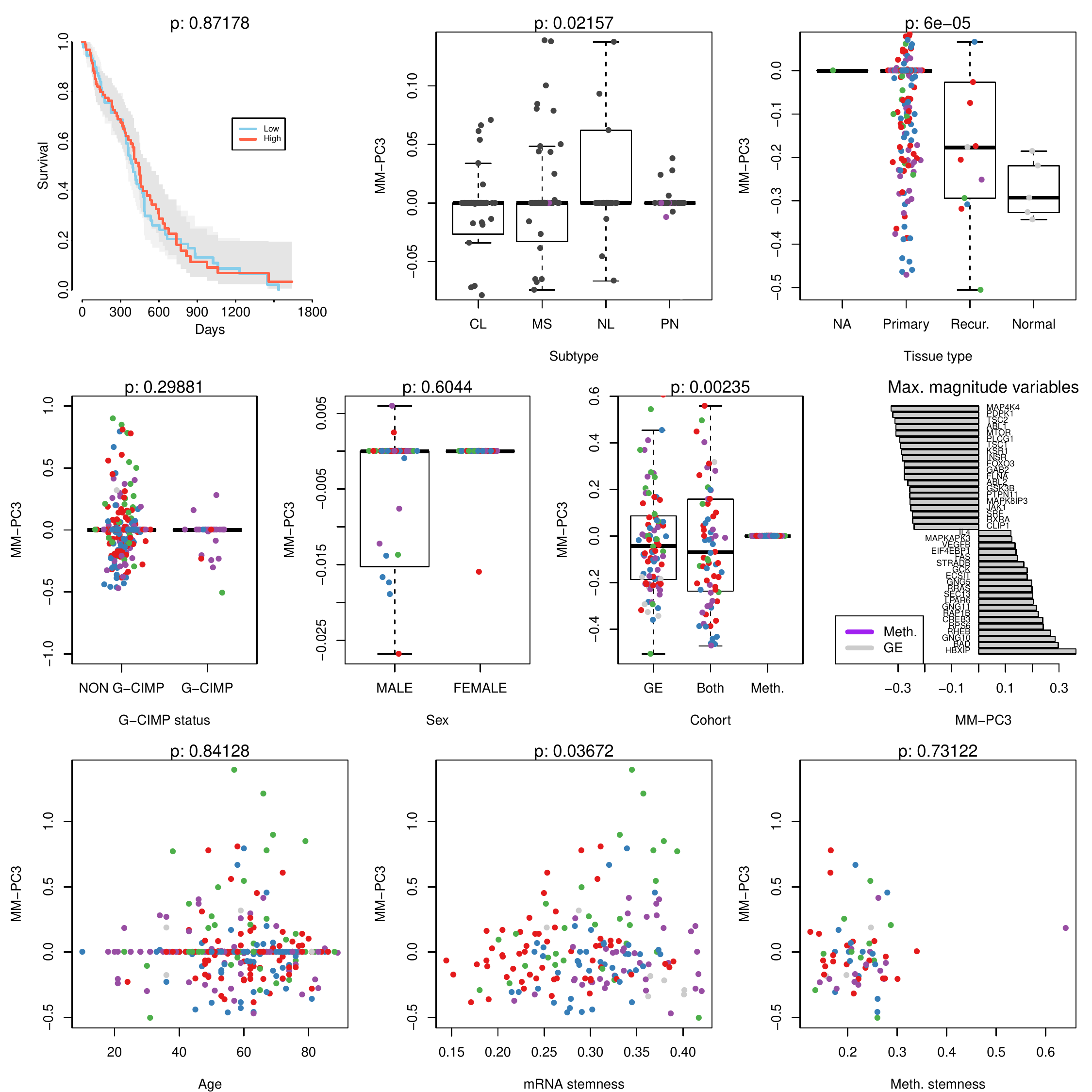}
  \caption{MM-PCA component 3. See caption of Figure \ref{fig:compfirst}.}
\end{figure}
\begin{figure}
  \centering
    \includegraphics[width=\textwidth]{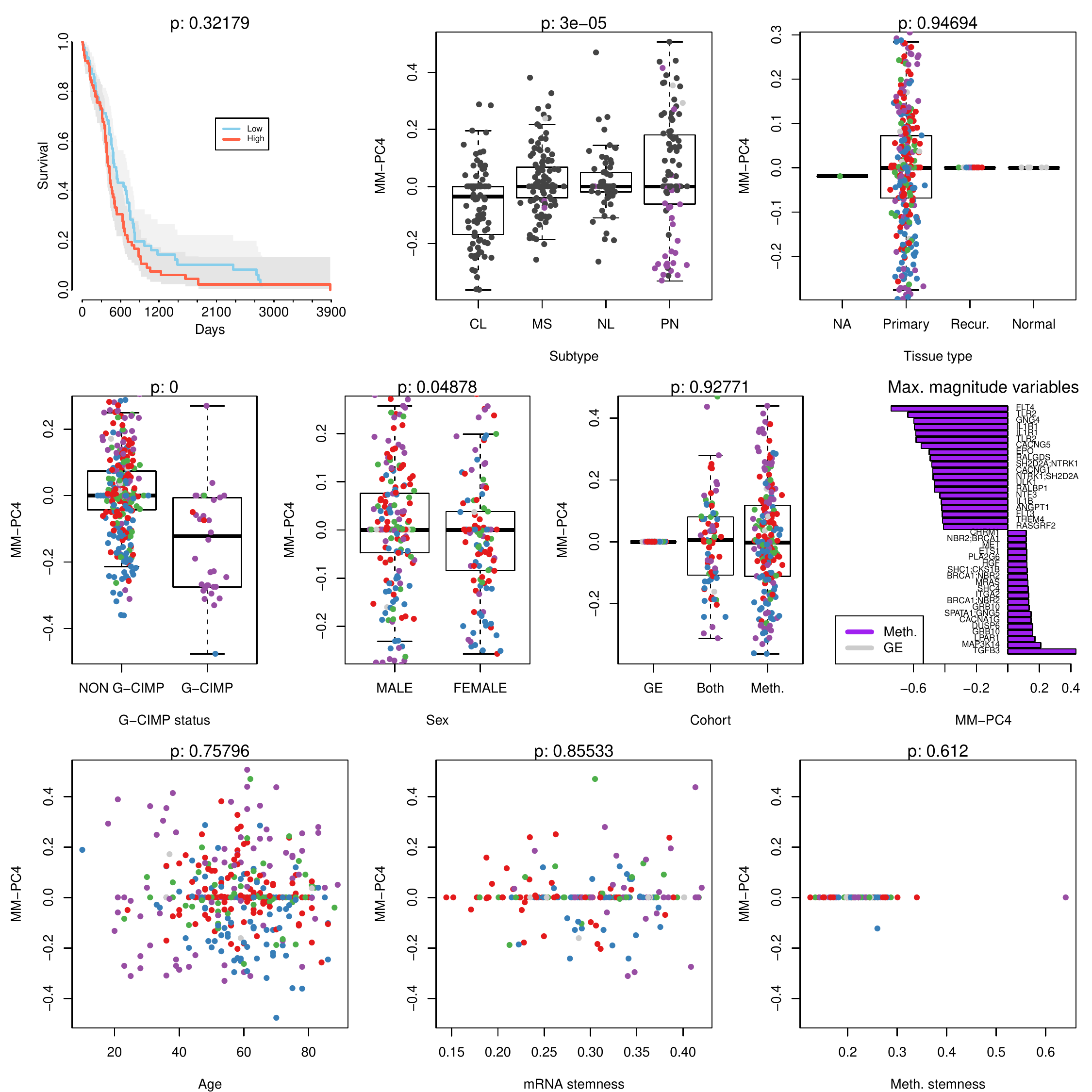}
  \caption{MM-PCA component 4. See caption of Figure \ref{fig:compfirst}.}
\end{figure}
\begin{figure}
  \centering
    \includegraphics[width=\textwidth]{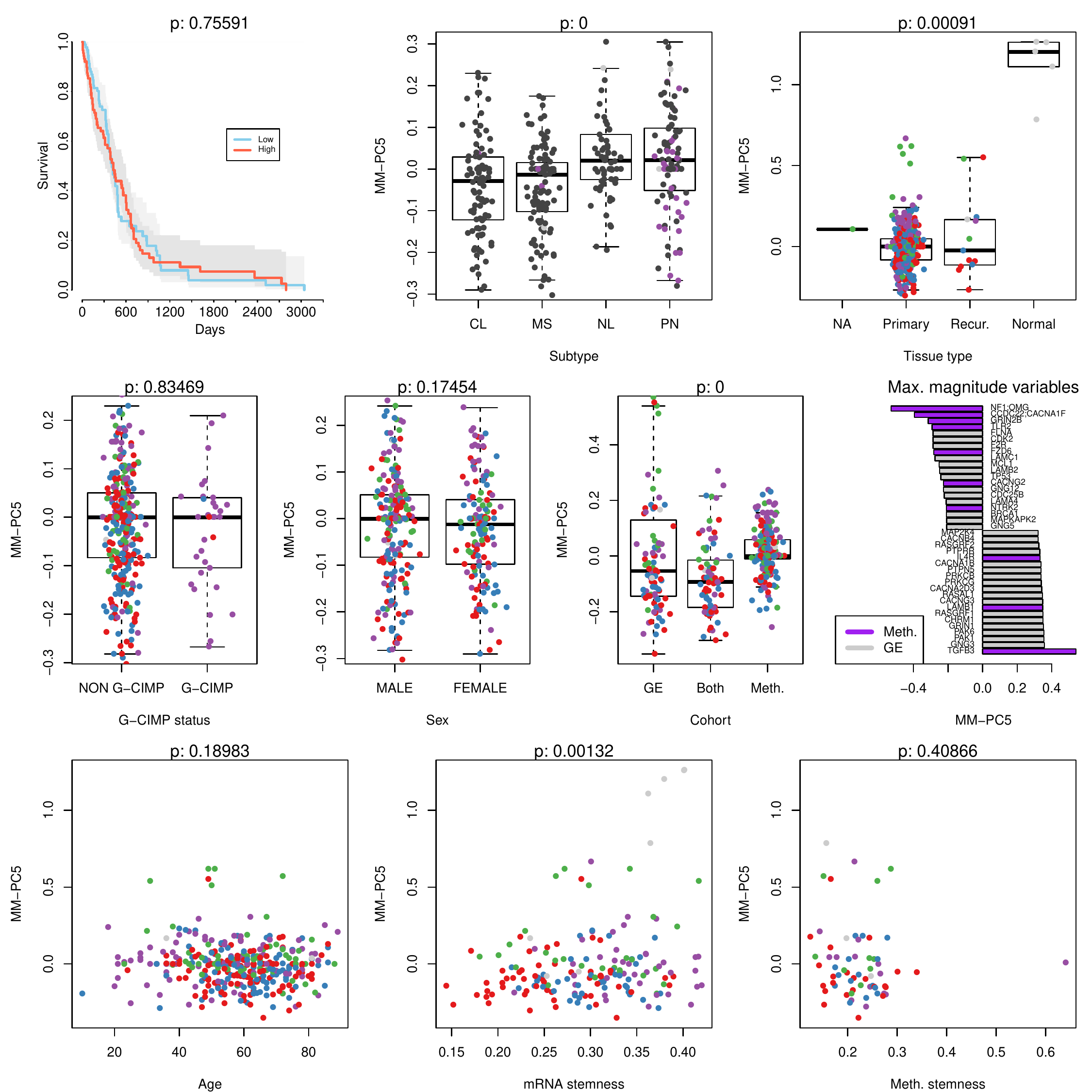}
  \caption{MM-PCA component 5. See caption of Figure \ref{fig:compfirst}.}
\end{figure}
\begin{figure}
  \centering
    \includegraphics[width=\textwidth]{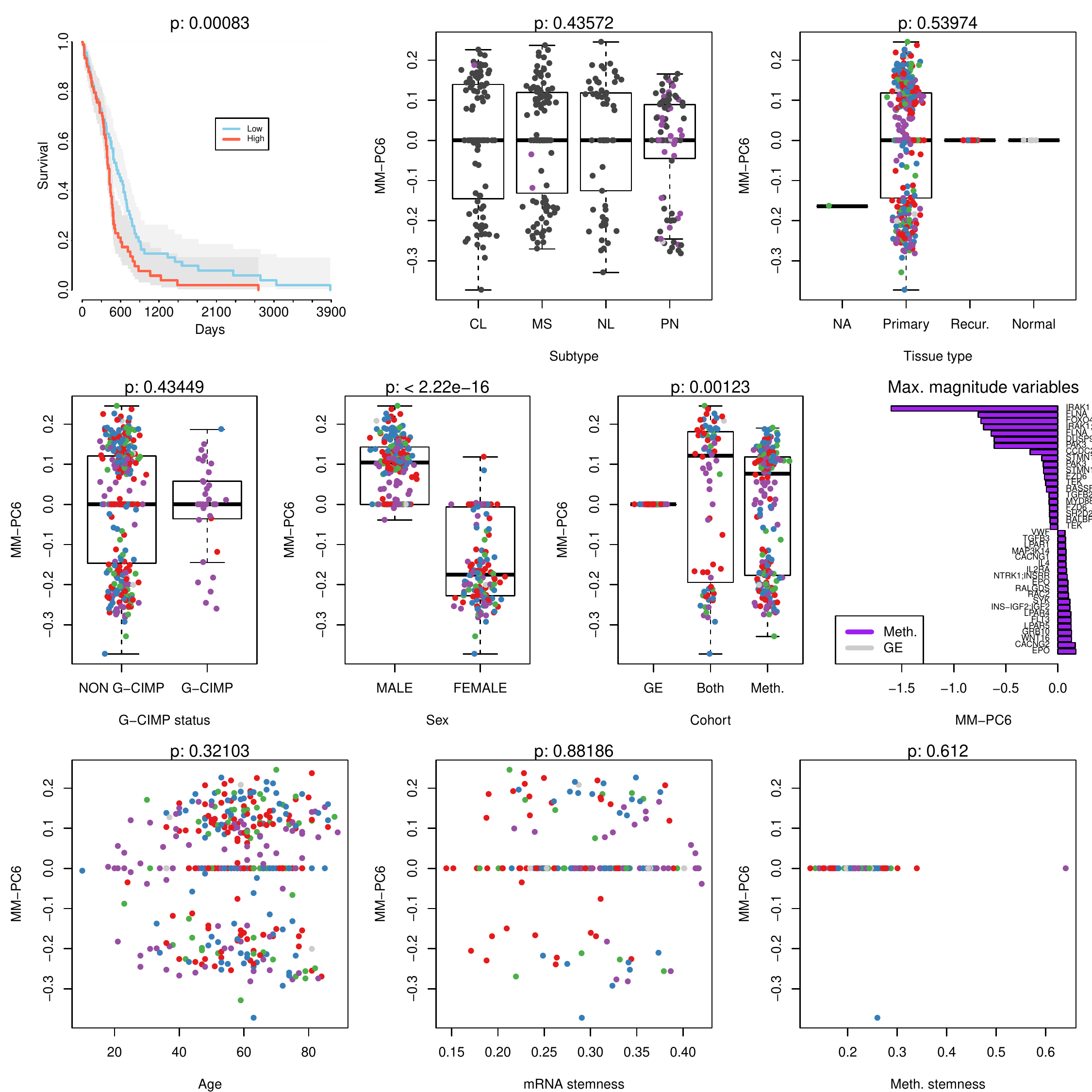}
  \caption{MM-PCA component 6. See caption of Figure \ref{fig:compfirst}.}
  \label{fig:complast}
\end{figure}
\begin{figure}
  \centering
    \includegraphics[width=\textwidth]{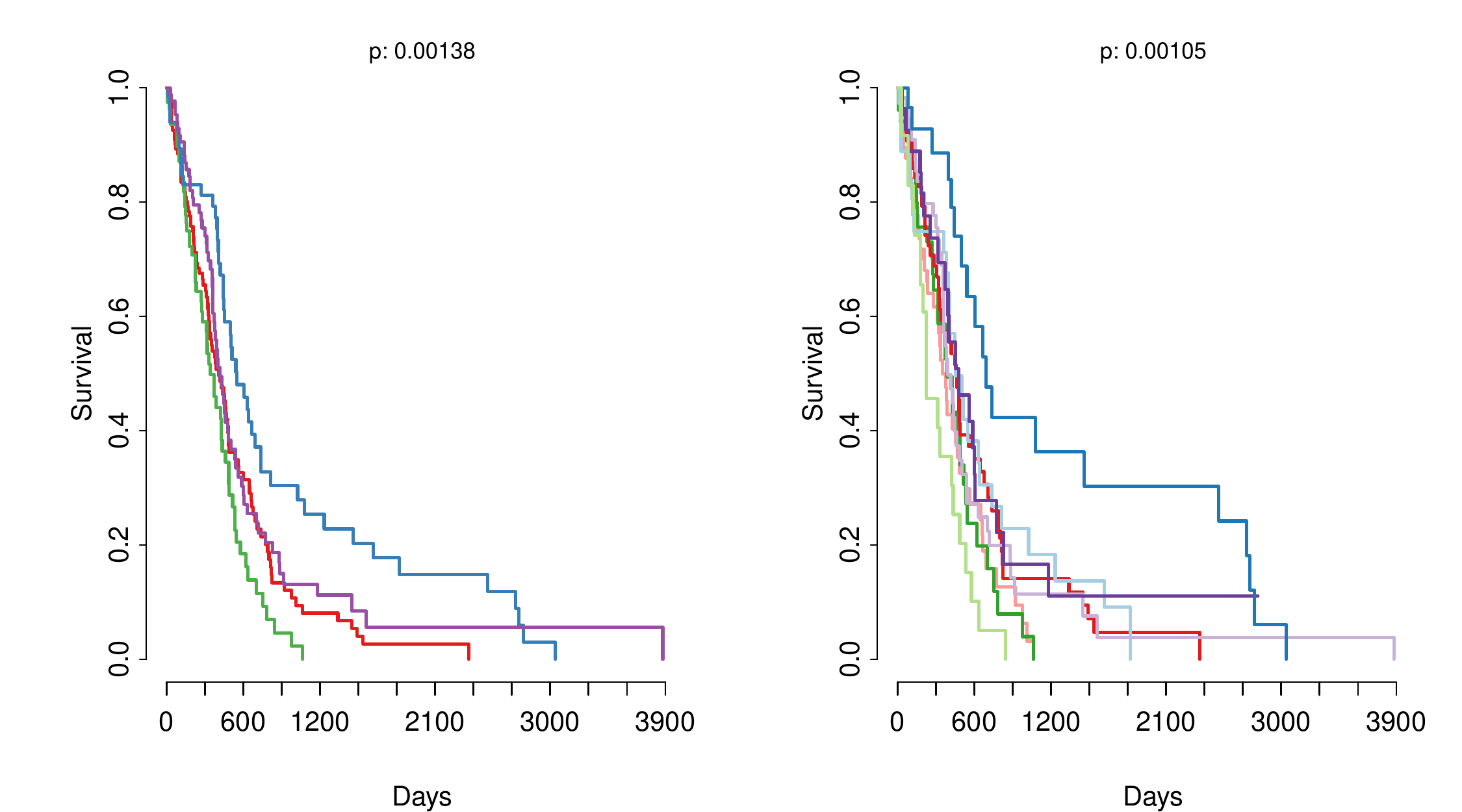}
  \caption{Survival plots for patients stratified by MM-PCA bi-clustering. The first plot uses cluster assignments based on two components. The second plots uses cluster assignments based on three components. Colors refer to cluster colors from Figure \ref{fig:bioint}A. P-values are based on Cox regression for patients in the (dark) blue cluster versus all other patients.}
  \label{fig:survclust}
\end{figure}

\end{document}